\documentclass[journal,onecolumn,draftclsnofoot]{IEEEtran}
\usepackage{graphicx} 

\usepackage{pgfplots}
\usepgfplotslibrary{groupplots}
\pgfplotsset{compat=1.18}

\usepackage{amsmath, amsfonts, amssymb, amsthm}
\usepackage{mathtools}
\usepackage{xcolor,hyperref}
\usepackage[capitalize]{cleveref}
\usepackage{bm}
\usepackage{xfrac}
\usepackage{url}
\usepackage{algorithm}
\usepackage{algpseudocodex}
\usepackage{xfrac}
\usepackage{bbm}

\theoremstyle{plain}
\newtheorem{lemma}{Lemma}
\newtheorem{theorem}{Theorem}
\newtheorem{corollary}{Corollary}
\newtheorem{proposition}{Proposition}

\theoremstyle{definition}
\newtheorem{definition}{Definition}
\newtheorem{example}{Example}

\theoremstyle{remark}
\newtheorem{remark}{Remark}

\newcommand{\CC}{\mathbb{C}}
\newcommand{\EE}{\mathbb{E}}
\newcommand{\FF}{\mathbb{F}}
\newcommand{\RR}{\mathbb{R}}
\newcommand{\ZZ}{\mathbb{Z}}

\newcommand{\C}{\mathcal{C}}

\renewcommand{\Re}{\operatorname{Re}}
\renewcommand{\Im}{\operatorname{Im}}
\newcommand{\imag}{\mathbbm{i}}

\renewcommand{\sec}{\mathrm{sec}}
\newcommand{\enc}{\mathrm{enc}}

\newcommand{\CN}{\mathcal{CN}}

\DeclareMathOperator{\Var}{Var}

\DeclareMathOperator{\tr}{tr}
\DeclareMathOperator{\diag}{diag}

\newcommand\independent{\protect\mathpalette{\protect\independenT}{\perp}}
\def\independenT#1#2{\mathrel{\rlap{$#1#2$}\mkern2mu{#1#2}}}

\title{Analog~Secure~Distributed~Matrix~Multiplication}
\author{Okko Makkonen, and Camilla Hollanti \thanks{The authors are with the Department of Mathematics and Systems Analysis, Aalto University, Espoo, Finland. (emails: okko.makkonen@aalto.fi, camilla.hollanti@aalto.fi)}
\thanks{Preliminary results were published in the Proceedings of the 2022 IEEE Internatioal Symposium on Information Theory \cite{makkonen2022analog}.}
\thanks{This work was partially supported by the Wallenberg AI, Autonomous Systems and Software Program. The work of O.~Makkonen is supported by the Vilho, Yrjö and Kalle Väisälä Foundation of the Finnish Academy of Science and Letters and by the Finnish Foundation for Technology Promotion.}
}
\date{\today}

\begin{document}

\maketitle

\begin{abstract}
In this paper, we present secure distributed matrix multiplication (SDMM) schemes over the complex numbers with good numerical stability and small mutual information leakage by utilizing polynomial interpolation with roots of unity. Furthermore, we give constructions utilizing the real numbers by first encoding the real matrices to smaller complex matrices using a technique we call \emph{complexification}. These schemes over the real numbers enjoy many of the benefits of the schemes over the complex numbers, including good numerical stability, but are computationally more efficient. To analyze the numerical stability and the mutual information leakage, we give some bounds on the condition numbers of Vandermonde matrices whose evaluation points are roots of unity.
\end{abstract}

\section{Introduction}

Secure distributed matrix multiplication (SDMM) refers to outsourcing the computation of a matrix product to worker nodes without revealing the contents of the matrices. SDMM was first introduced in \cite{chang2018capacity} and has been widely studied in, e.g.,  \cite{kakar2019capacity, d2020gasp, mital2022secure, byrne2023straggler, karpuk2024modular, machado2023hera, machado2022root}. The schemes presented in the literature utilize methods from coding theory to provide robustness and error tolerance, while the security of the schemes is based on secret sharing.

SDMM has been considered over finite fields due to the techniques in coding theory and secret sharing. However, applications of matrix multiplication in machine learning and artificial intelligence, as well as in data science and engineering, require computations with real valued data. Utilizing finite fields in this setting is possible by discretizing the real valued data to integers and performing the computations modulo a large prime number. This is a fixed point representation of the numbers as the numbers are represented with a fixed number of precision after the decimal point. One has to be careful to choose the prime number to be sufficiently large such that the computations do not ``overflow'', which would cause catastrophic errors.

As is usual in numerical computations, we wish to utilize \emph{floating point numbers}, which efficiently trades off magnitude and precision. They also have efficient hardware implementations compared to finite field arithmetic, which makes them attractive from an implementation point of view
Therefore, our aim is to construct SDMM schemes over the field of real numbers and the field of complex numbers such that the computations can be performed using floating point numbers. We refer to the real or complex numbers as the \emph{analog domain} as opposed to the discrete finite fields.
As is usual with numerical computations, one has to be careful to design numerically stable algorithms, since small rounding errors in the computations can otherwise accumulate.

\subsection{Contributions and Related Work}

Computing over the real numbers has raised the interest of the coded computing community. For example, analog distributed matrix multiplication without the security aspect was considered in \cite{ramamoorthy2022numericallya}. On the other hand, secure computation was considered in \cite{soleymani2021analog}, where the authors extended the Lagrange coded computing framework in \cite{yu2019lagrange} to the analog domain. Our preliminary conference article \cite{makkonen2022analog} considered secure distributed matrix multiplication over the complex numbers, but did not present constructions utilizing the real numbers. This paper extends the ideas in \cite{makkonen2022analog}, which now corresponds to \cref{sec:SDMM_over_complex_numbers}, and provides a much more thorough analysis as described in detail below.

In this paper, we present SDMM schemes over the complex numbers with good numerical stability and small mutual information leakage. In particular, we convert various schemes based on polynomial encoding to the complex domain by using roots of unity as evaluation points. Furthermore, we give constructions utilizing the real numbers by first encoding the real matrices to smaller complex matrices using a technique we call \emph{complexification}, which is similar to the rotation matrix embedding used in \cite{ramamoorthy2022numericallya}. These schemes over the real numbers enjoy many of the benefits of the schemes over the complex numbers, including good numerical stability, but are computationally more efficient. To analyze the numerical stability and the mutual information leakage, we give some bounds on the condition numbers of Vandermonde matrices whose evaluation points are roots of unity. These results may be of independent interest.

\subsection{Organization}

This paper is organized as follows. In \cref{sec:preliminaries}, we give preliminaries on numerical stability, information theory, and secret sharing. In \cref{sec:linear_SDMM}, we go over the basics of linear SDMM schemes. In \cref{sec:SDMM_over_complex_numbers}, we present our SDMM constructions over the complex numbers, which utilize polynomial interpolation with roots of unity as evaluation points. In \cref{sec:SDMM_over_real_numbers}, we present an encoding scheme to utilize the techniques over the complex numbers by first encoding the real matrices as smaller complex matrices. This technique allows us to reduce the communication and computation costs compared to simply embedding the real matrices as complex matrices. In \cref{sec:numerical_experiments}, we present numerical experiments that verify the numerical accuracy of our methods as a function of the information leakage. Finally, in \cref{sec:appendix} we derive some bounds on the condition numbers of Vandermonde matrices whose evaluation points are some subset of the $n$th roots of unity.

\section{Preliminaries}\label{sec:preliminaries}

We denote the set $\{1, 2, \dots, n\}$ by $[n]$. We will denote an arbitrary field with $\FF$ and a finite field with $q$ elements by $\FF_q$. Vectors are either row or column vectors depending on the context. We denote a vector norm and the induced matrix norm by $\lVert \cdot \rVert$. We will consider mainly the 2-norm, the $\infty$-norm, and the Frobenius norm and write these as $\lVert \cdot \rVert_2$, $\lVert \cdot \rVert_\infty$, and $\lVert \cdot \lVert_F$, respectively, if we want to emphasize the particular norm. We denote $[n]$ for the set $\{1, \dots, n\}$. Laurent polynomials over the field $\FF$ are polynomials in $\FF[z, z^{-1}]$, i.e., they are univariate polynomials that are allowed to have negative exponents. We define the evaluation of a Laurent polynomial $f$ at $\alpha \neq 0$ in the obvious way and write $f(\alpha) \in \FF$. We write $f(z)$ to emphasize that $f$ is a Laurent polynomial in the variable $z$. If $f \in \CC[z, z^{-1}]$ is a Laurent polynomial over the complex numbers, then $\overline{f} \in \CC[z, z^{-1}]$ denotes the coefficientwise conjugation.

Let $\alpha_1, \alpha_2, \dots, \alpha_n \in \FF$ be distinct elements and $\gamma_1, \gamma_2, \dots, \gamma_k \in \ZZ$ be distinct integers. A $k \times n$ matrix $V$ with entries $V_{ji} = \alpha_i^{\gamma_j}$ is said to be a \emph{generalized Vandermonde matrix} over the field $\FF$. The elements $\alpha_i$ are the \emph{evaluation points} and the elements $\gamma_j$ are the \emph{exponents}. Some authors only allow non-negative exponents, but we will make no distinction between positive or negative exponents. Such a matrix with exponents $\gamma_j = j - 1$ is said to be a \emph{Vandermonde matrix}.

We write the imaginary unit as $\imag$, i.e., $\imag^2 = -1$. By writing $(a + \imag b)(c + \imag d) = (ac - bd) + \imag(ad + bc)$, the product of two complex numbers can be computed with 4 real multiplications and 2 additions. On the other hand, by writing
\begin{equation}\label{eq:karatsuba_algorithm_complex_numbers}
    (a + \imag b)(c + \imag d) = (ac - bd) + \imag((a + c)(b + d) - ac - bd)
\end{equation}
the computation can be done with 3 real multiplications and 5 additions. This method is called the 3M method, see \cite[Chapter~22]{higham2002accuracy}. If multiplications are more expensive than additions, like in the case of matrices, it may be beneficial to utilize this method. On the other hand, if multiplications and additions are equally expensive, then the naive algorithm only uses 6 operations compared to 8 in \eqref{eq:karatsuba_algorithm_complex_numbers}.

\subsection{Numerical Stability and Condition Numbers}\label{sec:numerical_stability}

In this section, we review some concepts from numerical analysis, such as the sources of errors in numerical computations and the accumulation of these errors. For a comprehensive reference on the topic, see \cite{higham2002accuracy}.

Representing continuous values on a computer is often done with a floating point representation as it has a good trade-off between precision and being able to represent numbers of varying magnitudes. Additionally, floating point arithmetic has great implementation support on modern computer hardware, which makes it preferable for many applications. However, due to the finite precision that is available with floating point numbers, errors are bound to happen in computations. In particular, after the computation of an arithmetic operation, the result has to be rounded back to a floating point number. The maximal relative error in the rounding is known as the \emph{unit roundoff} and denoted by $u$. For example, for the IEEE 754 single-precision floating point numbers, $u \approx 5.96 \cdot 10^{-8}$.

Performing a computation consisting of multiple operations will yield many round-off errors and the total error in the result is an accumulation of the (small) individual errors. For example, computing an inner product of vectors $x, y \in \RR^n$ will yield an error of at most $\gamma_n |x|^T |y|$, where $\gamma_n = \frac{nu}{1 - nu}$ and $|\cdot|$ denotes the componentwise absolute value. Hence, it is not possible to compute the product to high relative accuracy since $|x|^T |y|$ can be significantly larger than $|x^T y|$.
In addition to inner products, the accuracy of matrix multiplication can be analyzed similarly. A product of matrices $A$ and $B$ can be computed to an absolute error proportional to $\lVert A \rVert \cdot \lVert B \rVert$.

The inherent sensitivity to errors is measured by the condition number of a computational problem. Consider the system of linear equations given by $Ax = b$ for some matrix $A$ and vector $b$. Let us assume that we know the value of $A$ exactly, but instead of $b$ we have $b + \delta b$, where $\frac{\lVert \delta b \rVert}{\lVert b \rVert}$ is some relative error. Given this, we are interested in bounding the relative error in the solution $x + \delta x$. The condition number of a matrix $A$ with respect to the given norm is defined as $\kappa(A) = \lVert A \rVert \lVert A^{-1} \rVert$. Then, the relative error in the solution is
\begin{equation*}
    \frac{\lVert \delta x \rVert}{\lVert x \rVert} \leq \kappa(A) \frac{\lVert \delta b \rVert}{\lVert b \rVert}.
\end{equation*}
Let $\kappa(A)$ and $\kappa'(A)$ be condition numbers of a matrix with respect to two different norms. As all norms on a finite dimensional vector spaces are equivalent, there exists a constant $c$ that only depends on the norms such that $\kappa(A) \leq c \kappa'(A)$ for all $A$. A rule-of-thumb is that we lose $k$ decimal digits of precision if $\kappa(A) \approx 10^k$.

One problem that can be solved by a system of linear equations is polynomial interpolation. Let $f \in \FF[x]$ be a polynomial of degree $< m$ written as $f(x) = f_0 + f_1 x + \dots + f_{m-1} x^{m-1}$. Given $m$ evaluations of $f$ at distinct points $\alpha_1, \dots, \alpha_m$, it is possible to decode the coefficients of $f$ with polynomial interpolation. In particular, we may write
\begin{equation*}
    \begin{pmatrix}
        f(\alpha_1) \\
        f(\alpha_2) \\
        \vdots \\
        f(\alpha_m)
    \end{pmatrix} =
    \begin{pmatrix}
        1 & \alpha_1 & \cdots & \alpha_1^{m-1} \\
        1 & \alpha_2 & \cdots & \alpha_2^{m-1} \\
        \vdots & \vdots & \ddots & \vdots \\
        1 & \alpha_m & \cdots & \alpha_m^{m-1}
    \end{pmatrix}
    \begin{pmatrix}
        f_0 \\
        f_1 \\
        \vdots \\
        f_{m-1}
    \end{pmatrix}
\end{equation*}
The coefficients can be obtained from the evaluations by solving the above system. The matrix on the right is (the transpose of) a Vandermonde matrix with evaluation points $\alpha_1, \dots, \alpha_m$. Therefore, the condition number of polynomial interpolation is determined by the condition number of the corresponding Vandermonde matrix whose evaluation points are $\alpha_1, \dots, \alpha_m$. Similarly, if $f \in \FF[z, z^{-1}]$ written as $f(z) = f_{-\ell} z^{-\ell} + f_{-\ell + 1} z^{-\ell + 1} + \cdots + f_0 + f_1 z + \cdots f_k z^k$, then we can solve the coefficients from $m = k + \ell - 1$ evaluations by considering the polynomial $z^\ell f(z)$.

It is known that the condition number of a Vandermonde matrix with real evaluation points grows exponentially as a function of $m$, see \cite{gautschi1987lower}. On the other hand, if the evaluation points are chosen to be $m$th roots of unity in the complex numbers, the condition number achieves the lowest possible value of 1. Low condition numbers can also be obtained if the evaluation points are more or less equally spaced on the unit circle in $\CC$, see \cite{pan2016how}. Indeed, it was shown in \cite{ramamoorthy2022numericallya} that the condition number of an $m \times m$ Vandermonde matrix, whose evaluation points are $n$th roots of unity for $n = m + d$ for some constant $d$, grows polynomially in $n$. In particular, it was shown that $\kappa_2(V) = \mathcal{O}(n^{d + 5.5})$. In \cref{sec:appendix}, we improve this bound to $\kappa_2(V) = \mathcal{O}(n^{d + 2})$, see \cref{thm:large_vandermonde_inverse_2_norm_upper_bound}. We also show that the norm of the inverse of an $m \times m$ Vandermonde matrix with $n$th roots of unity as evaluation points is $\mathcal{O}(n^{m - 1})$ if $m$ is constant, see \cref{thm:small_vandermonde_inverse_2_norm_upper_bound} and \cref{cor:small_vandermonde_inverse_frobenius_norm_upper_bound}. We use these results later to show the security of our proposed schemes and to analyze the numerical stability.

\subsection{Information Theory and Channel Capacity}\label{sec:information_theory}

The mutual information between two random variables is denoted by $I(\bm{x}; \bm{y})$ and the differential entropy and conditional entropy by $h(\bm{x})$ and $h(\bm{x} \mid \bm{y})$, respectively. The covariance matrix of a random (column) vector $\bm{x}$ with expected value $\mu = \EE[\bm{x}]$ is $\Var(\bm{x}) = \EE[(\bm{x} - \mu)(\bm{x} - \mu)^*]$, often denoted by $\Sigma$. We say that a complex random variable $\bm{z}$ is circularly-symmetric if $e^{\imag \theta} \bm{z}$ and $\bm{z}$ are identically distributed for all $\theta \in \RR$. The circularly-symmetric (central) complex normal distribution with covariance $\Sigma$ is denoted by $\CN(0, \Sigma)$. It turns out that this distribution maximizes the differential entropy of a random vector.

\begin{proposition}[{\cite[Theorem 2]{neeser1993proper}}]\label{prop:entropy_upper_bound}
Let $\bm{z}$ be a continuous complex random vector of length $n$ with a nonsingular covariance matrix $\Sigma$. Then $h(\bm{z}) \leq \log((\pi e)^n \det(\Sigma))$, with equality if and only if $\bm{z}$ is a circularly-symmetric complex normal random vector.
\end{proposition}

The following standard lemma essentially determines the capacity of the multiple-input multiple-output (MIMO) additive white Gaussian noise (AWGN) channel. We include the proof here for completeness.

\begin{lemma}\label{lem:AWGN_MIMO_capacity}
Let $\bm{x} \in \CC^m$ and $\bm{z} \in \CC^n$ be independent random variables with $Q = \Var(\bm{x})$ and $\bm{z} \sim \CN(0, \sigma^2 \cdot I_n)$. Let $H \in \CC^{n \times m}$ and define the random variable $\bm{y} = H \bm{x} + \bm{z} \in \CC^n$. Then
\begin{equation*}
    I(\bm{x}; \bm{y}) \leq \log(\det(I_n + \tfrac{1}{\sigma^2} H Q H^*)).
\end{equation*}
\end{lemma}

\begin{proof}
By using the definition of mutual information and the upper bound for the differential entropy in \cref{prop:entropy_upper_bound} we get that
\begin{align*}
    I(\bm{x}; \bm{y}) &= h(\bm{y}) - h(\bm{z}) \\
    &\leq \log((\pi e)^n \det(\Var(\bm{y}))) - \log((\pi e)^n \det(\Var(\bm{z}))) \\
    &= \log \left( \frac{\det(\Var(\bm{y}))}{\det(\Var(\bm{z}))} \right).
\end{align*}
By independence of $H\bm{x}$ and $\bm{z}$, we obtain
\begin{align*}
    \det(\Var(\bm{y})) &= \det(\Var(H\bm{x}) + \Var(\bm{z})) \\
    &= \det(HQH^* + \sigma^2 I_n).
\end{align*}
Additionally, $\det(\Var(\bm{z})) = \det(\sigma^2 \cdot I) = (\sigma^2)^n$, so the result follows.
\end{proof}

The above bound on the mutual information still depends on the covariance matrix of the input $\bm{x}$. By imposing the power-constraint $\tr(Q) \leq P$, we obtain the following result, which works well for $\sigma^2 \gg 1$. The below uses properties of the natural logarithm, so the mutual information is measured in base $e$, i.e., in the unit \emph{nat}.

\begin{theorem}\label{thm:AWGN_MIMO_power_constrained_capacity}
Let $\bm{x} \in \CC^m$ and $\bm{z} \in \CC^n$ be independent random variables with $Q = \Var(\bm{x})$ and $\bm{z} \sim \CN(0, \sigma^2 \cdot I_n)$. Let $H \in \CC^{n \times m}$ and define the random variable $\bm{y} = H \bm{x} + \bm{z} \in \CC^n$. If $\tr(Q) \leq P$, then
\begin{equation*}
    I(\bm{x}; \bm{y}) \leq \tfrac{P}{\sigma^2} \lVert H \rVert_F^2.
\end{equation*}
\end{theorem}

\begin{proof}
If $A$ is a positive semi-definite matrix, then $\log(\det(I + A)) \leq \tr(A)$, which is easy to check by considering the eigenvalues. The matrix $HQH^*$ is clearly positive semi-definite. Then, by \cref{lem:AWGN_MIMO_capacity}, and the submultiplicativity of the trace on positive semi-definite matrices
\begin{align*}
    I(\bm{x}; \bm{y}) &\leq \log(\det(I_n + \tfrac{1}{\sigma^2} HQH^*)) \\
    &\leq \tfrac{1}{\sigma^2} \tr(HQH^*) \\
    &\leq \tfrac{1}{\sigma^2} \tr(HH^*) \tr(Q).
\end{align*}
The result follows as $\tr(HH^*) = \lVert H \rVert_F^2$.
\end{proof}

\subsection{Secret Sharing}

Secret sharing is a method of distributing a secret value $s$ among $n$ participants such that the secret can be decoded from the shares held by the participants, but such that the shares of a small number of participants reveals no information about the secret value. In other words, if an adversary is able to obtain the shares from a only few participants, they will not be able to reveal the secret value. If any at most $t$ shares reveal no information about the secret value, then we say that the secret sharing scheme is \emph{$t$-secure}. For an overview on secret sharing schemes, see \cite{beimel2011secret}. We will be interested in secret sharing schemes obtained from linear codes and their generator matrices.

\begin{definition}\label{def:nested_coset_coding_scheme}
A \emph{nested coset coding scheme} over the field $\FF$ consists of a pair $(G^\enc, G^\sec)$ of full rank matrices $G^\enc \in \FF^{m \times n}$ and $G^\sec \in \FF^{k \times n}$ whose row spaces intersect trivially. The secret message $s \in \FF^m$ is encoded to the shares as $\hat{s} = sG^\enc + rG^\sec \in \FF^n$, where $r \in \FF^k$ is chosen from a suitable distribution. Party $i \in [n]$ gets the coordinate $\hat{s}_i$.
\end{definition}

This construction has been seen in the literature in, e.g., \cite[Section~4.2]{chen2007secure}, but it is usually presented in terms of the codes $(\C^\enc \oplus \C^\sec, \C^\sec)$, where $\C^\enc, \C^\sec$ are generated by $G^\enc, G^\sec$, respectively. We choose to keep track of the generator matrices instead of the codes so that we have a fixed basis. In the secret sharing literature, the above is also a special case of a multi-target monotone span program, see the survey \cite{beimel2011secret}.

Let $(G^\enc, G^\sec)$ define a nested coset coding scheme. As the row spaces of $G^\enc$ and $G^\sec$ intersect trivially, the vector $\hat{s}$ can be uniquely decomposed as the sum of $s G^\enc$ and $r G^\sec$ and the secret value $s$ can be obtained from $\hat{s}$. On the other hand, an adversary will obtain some coordinates of $\hat{s}$, say those indexed by $\mathcal{T} \subset [n]$ with $|\mathcal{T}| \leq t$. What the adversary sees is $\hat{s}_\mathcal{T} = s G^\enc_\mathcal{T} + r G^\sec_\mathcal{T} \in \FF^t$. This can be seen as the output of a noisy communication channel with input $s$ and output $\hat{s}_\mathcal{T}$. The goal of designing good secret sharing schemes is the opposite of designing good communication channels, since we want to \emph{minimize} the mutual information between the input and the output.

If $\FF = \FF_q$ is a finite field, we may choose $r \in \FF_q^k$ uniformly at random. If $k = t$ and $G^\sec_\mathcal{T}$ is invertible for all $\mathcal{T} \subset [n]$ with $|\mathcal{T}| = t$ (i.e. if $G^\sec$ has the MDS property), then $\hat{s}_\mathcal{T}$ is uniformly distributed and reveals no information about the secret value $s$. In general, we require that the dual minimum distance of the code generated by $G^\sec$ is at least $t + 1$ since then any $t$ columns of $G^\sec$ are linearly independent.

\begin{example}[Shamir secret sharing]
The well-known Shamir secret sharing scheme can be described by using a generator matrix
\begin{equation*}
    G = \begin{pmatrix}
        1 & 1 & \cdots & 1 \\
        \alpha_1 & \alpha_2 & \dots & \alpha_n \\
        \vdots & \vdots & \ddots & \vdots \\
        \alpha_1^{m + t - 1} & \alpha_2^{m + t - 1} & \cdots & \alpha_n^{m + t - 1}
    \end{pmatrix} = \begin{pmatrix}
        G^\enc \\
        G^\sec
    \end{pmatrix}.
\end{equation*}
The shares are the coordinates of $\hat{s} = (s, r) G = s G^\enc + r G^\sec$, where $G^\enc$ and $G^\sec$ consist of the first $m$ and last $t$ rows of $G$, respectively. As long as $\alpha_1, \dots, \alpha_n \neq 0$ and the evaluation points $\alpha_1, \dots, \alpha_n$ are pairwise distinct, the matrix $G^\sec$ has the MDS property. This means that the Shamir secret sharing scheme is $t$-secure. Furthermore, as the vector of shares $\hat{s}$ is a codeword in a Reed\nobreakdash--Solomon code of dimension $m + t$, it is clear that any $m + t$ shares are sufficient to decode the secret $s$.
\end{example}

If the field is $\FF = \RR$ or $\FF = \CC$, we can not choose the entries in the random vector $r$ uniformly at random anymore which means that we can not achieve perfect information theoretic privacy. Instead, we measure the privacy of such a scheme by the information leakage.

\begin{definition}[$(t, \delta)$-security]
Assume that an adversary has access to the shares of the parties indexed by $\mathcal{T} \subset [n]$. We measure the \emph{average leakage per input symbol} of a secret sharing scheme as $\mathcal{L}(\bm{s} \to \hat{\bm{s}}_\mathcal{T}) = \frac{1}{m}I(\bm{s}; \hat{\bm{s}}_\mathcal{T})$, where the secret $\bm{s}$ consists of $m$ symbols. Let $\delta > 0$ denote the tolerated amount of leakage. We say that a secret sharing scheme is \emph{$(t, \delta)$\nobreakdash-secure} if
\begin{equation*}
    \max_{\substack{\mathcal{T} \subset [n] \\ |\mathcal{T}| = t}} \mathcal{L}(\bm{s} \to \hat{\bm{s}}_\mathcal{T}) \leq \delta
\end{equation*}
for all distributions of $\bm{s}$ such that $\lVert \bm{s} \rVert_\infty \leq 1$, i.e., if the leakage is sufficiently small for all subsets of size $t$. Notice that $(t, 0)$-security matches with the definition of $t$-security from above.
\end{definition}

As is the case with communication over noisy channels, a continuous channel can have infinite capacity if there are no constraints on the possible inputs. For this reason, we impose that $\lVert s \rVert_\infty$ is at most 1. The choice of the upper bound is arbitrary, but we choose the constant 1 to have fewer parameters in our formulas. We normalize the information leakage for the number of symbols $m$ to make comparisons simpler, though the parameter $\delta$ could be rescaled to remove this factor. Apart from the normalization term, our definition of security is the same as the semantic security in \cite{bloch2021overview}. We remark that we do not consider the security asymptotically as $m \to \infty$, so our definition does not coincide with the notion of weak secrecy.

The authors in \cite{tjell2021privacy} proposed choosing the random vector from a normal distribution with sufficiently large variance. We can recreate the Shamir secret sharing scheme over $\CC$ by choosing the evaluation points $\alpha_1, \dots, \alpha_n \in \CC$ as the $n$th roots of unity. The theorem below determines the required noise variance to achieve $(t, \delta)$-security. Notice that the computation of the required noise variance is simple and does not require maximization over a large number of cases as in \cite[Corollary~2]{soleymani2021analog}.

Consider the sequence defined by the product of the first $m$ entries of the sequence $1, 1, 2, 2, 3, 3, 4, 4, \dots$. It is simple to verify that the $m$th entry is given by $\Pi(m)$, where $\Pi$ is defined by $\Pi(2k) = k!^2$ and $\Pi(2k + 1) = k!^2 \cdot (k + 1)$. This sequence is A010551 on the OEIS \cite{oeis}. This sequence is not fundamental to the problem, but just falls from the analysis in \cref{cor:small_vandermonde_inverse_frobenius_norm_upper_bound}. The constant depending on $t$ could be improved with more thorough analysis.

\begin{theorem}\label{thm:secret_sharing_complex_numbers}
Let $(G^\enc, G^\sec)$ define a nested coset coding scheme over $\CC$, where $G^\enc \in \CC^{m \times n}$ and $G^\sec \in \CC^{t \times n}$ are generalized Vandermonde matrices with the $n$th roots of unity as evaluation points. Furthermore, assume that the exponents in $G^\sec$ are consecutive. By setting the noise variance as
\begin{equation*}
    \sigma^2 = \frac{1}{\delta} \cdot \frac{t^3 m}{ 4^{t - 1} \Pi(t - 1)^2} n^{2t - 2},
\end{equation*}
where $\Pi$ is defined above, the corresponding secret sharing scheme is $(t, \delta)$-secure.
\end{theorem}

\begin{proof}
Let $\mathcal{T} = \{i_1, \dots, i_t\} \subset [n]$ be a set of size $t$. The matrix $G^\enc_\mathcal{T}$ is an $m \times t$ generalized Vandermonde matrix with $n$th roots of unity as evaluation points. Therefore, $\lVert G^\enc_\mathcal{T} \rVert_F^2 = tm$. As the exponents in $G^\sec_\mathcal{T} \in \CC^{t \times t}$ are consecutive, say $\ell, \ell + 1, \dots, \ell + t - 1$, we may write
\begin{equation*}
    G^\sec_\mathcal{T} = V \cdot \diag(\alpha_{i_1}^\ell, \dots, \alpha_{i_t}^\ell),
\end{equation*}
where $V$ is a $t \times t$ Vandermonde matrix whose evaluation points are $\alpha_{i_1}, \dots, \alpha_{i_t}$. As $\diag(\alpha_{i_1}^\ell, \dots, \alpha_{i_t}^\ell)$ is unitary, we have that
\begin{equation*}
    \lVert (G^\sec_\mathcal{T})^{-1} \rVert_F^2 = \lVert V^{-1} \rVert_F^2 \leq \left( \frac{t}{2^{t - 1} \Pi(t - 1)} n^{t - 1} \right)^2
\end{equation*}
by \cref{cor:small_vandermonde_inverse_frobenius_norm_upper_bound}.

The adversary obtains a realization of $\hat{\bm{s}}_\mathcal{T} = \bm{s} G^\enc_\mathcal{T} + \bm{r} G^\sec_\mathcal{T}$, where $\bm{r} \sim \CN(0, \sigma^2 \cdot I)$. As this looks like the output of a Gaussian channel, we may use \cref{thm:AWGN_MIMO_power_constrained_capacity} to upper bound the mutual information $I(\bm{s}; \hat{\bm{s}}_\mathcal{T})$. In particular, $\hat{\bm{s}}_\mathcal{T}$ is essentially the output of a AWGN channel defined by $\bm{y} = H\bm{x} + \bm{z}$, where $H^T = G^\enc_\mathcal{T} (G^\sec_\mathcal{T})^{-1}$, $\bm{y}^T = \hat{\bm{s}}_\mathcal{T} (G^\sec_\mathcal{T})^{-1}$, $\bm{x}^T = \bm{s}$, and $\bm{z}^T = \bm{r}$. Let $\bm{s} \in \CC^m$ have an arbitrary distribution such that $\lVert \bm{s} \rVert_\infty \leq 1$. Then
\begin{equation*}
    \Var(\bm{s}_j) = \EE[|\bm{s}_j|^2] - |\EE[\bm{s}_j]|^2 \leq \EE[|\bm{s}_j|^2] \leq \EE[\lVert \bm{s} \rVert^2_\infty] \leq 1,
\end{equation*}
and $\tr(\Var(\bm{s})) = \sum_{j=1}^m \Var(\bm{s}_j) \leq m$. By utilizing \cref{thm:AWGN_MIMO_power_constrained_capacity}, we obtain
\begin{align*}
    \mathcal{L}(\bm{s} \to \hat{\bm{s}}_\mathcal{T}) &= \tfrac{1}{m} I(\bm{s}; \hat{\bm{s}}_\mathcal{T}) \\
    &\leq \tfrac{1}{m} \tfrac{m}{\sigma^2} \lVert G^\enc_\mathcal{T} (G^\sec_\mathcal{T})^{-1} \rVert_F^2 \\
    &\leq \tfrac{1}{\sigma^2} \cdot \lVert G^\enc_\mathcal{T} \rVert_F^2 \cdot \lVert (G^\sec_\mathcal{T})^{-1} \rVert_F^2 \\
    &\leq \frac{1}{\sigma^2} \cdot tm \cdot \left( \frac{t}{2^{t - 1} \Pi(t - 1)} n^{t - 1} \right)^2 \\
    &= \frac{1}{\sigma^2} \cdot \frac{t^3 m}{ 4^{t - 1} \Pi(t - 1)^2} n^{2t - 2} = \delta.
\end{align*}
As this holds for all sets $\mathcal{T} \subset [n]$ of size $t$, we get that the secret sharing scheme is $(t, \delta)$\nobreakdash-secure.
\end{proof}

Instead of secret sharing just the vector $s$, we could share many vectors using independent noise for each of them and give the $i$th share of each secret sharing vector to party $i$. The following lemma shows that the total leakage is upper bounded by the maximum leakage from any of the individual secret shares.

\begin{lemma}\label{lem:leakage_of_two_shares}
If $\bm{s}, \bm{s}'$ are shared independently as $\hat{\bm{s}}, \hat{\bm{s}}'$, then
\begin{equation*}
    \mathcal{L}((\bm{s}, \bm{s}') \to (\hat{\bm{s}}_\mathcal{T}, \hat{\bm{s}}'_\mathcal{T})) \leq \max\{\mathcal{L}(\bm{s} \to \hat{\bm{s}}_\mathcal{T}), \mathcal{L}(\bm{s}' \to \hat{\bm{s}}'_\mathcal{T})\}.
\end{equation*}
\end{lemma}

\begin{proof}
Clearly, due to creating the secret shares independently, $(\hat{\bm{s}}'_\mathcal{T} \independent \bm{s}) \mid \bm{s}'$ and $(\hat{\bm{s}}_\mathcal{T} \independent \bm{s}') \mid \bm{s}$, where $\cdot \independent \cdot \mid \cdot$ denotes conditional independence. Furthermore, $(\hat{\bm{s}}_\mathcal{T} \independent \hat{\bm{s}}'_\mathcal{T}) \mid (\bm{s}, \bm{s}')$ since the noise terms are chosen independently. Based on these conditional independences, we obtain the inequality
\begin{equation*}
    I(\bm{s}, \bm{s}'; \hat{\bm{s}}_\mathcal{T}, \hat{\bm{s}}'_\mathcal{T}) \leq I(\bm{s}; \hat{\bm{s}}_\mathcal{T}) + I(\bm{s}'; \hat{\bm{s}}'_\mathcal{T}).
\end{equation*}
Therefore, as $(\bm{s}, \bm{s}')$ contains $2m$ symbols,
\begin{align*}
    \mathcal{L}((\bm{s}, \bm{s}') \to (\hat{\bm{s}}_\mathcal{T}, \hat{\bm{s}}'_\mathcal{T})) &= \frac{1}{2m}I(\bm{s}, \bm{s}'; \hat{\bm{s}}_\mathcal{T}, \hat{\bm{s}}'_\mathcal{T}) \\
    &\leq \frac{1}{2}\left( \frac{1}{m} I(\bm{s}; \hat{\bm{s}}_\mathcal{T}) + \frac{1}{m} I(\bm{s}'; \hat{\bm{s}}'_\mathcal{T}) \right) \\
    &\leq \max\{\mathcal{L}(\bm{s} \to \hat{\bm{s}}_\mathcal{T}), \mathcal{L}(\bm{s}' \to \hat{\bm{s}}'_\mathcal{T}) \}. \qedhere
\end{align*}
\end{proof}

\section{Linear Secure Distributed Matrix Multiplication}\label{sec:linear_SDMM}

The goal of secure distributed matrix multiplication (SDMM) is to outsource the computation of a matrix product to $N$ worker nodes in such a way that the contents of the matrices are not revealed to adversaries in the system. In particular, we assume that an adversary has access to the data from some $X$ worker nodes. We say that $X$ is the \emph{security parameter}. Additionally, we assume that some of the worker nodes are \emph{stragglers}, meaning that they do not respond in time. We wish that the computation is not slowed down by the presence of these stragglers, which means that we need to be able to decode without these responses.

The general outline of an SDMM protocol is the following. The user has two matrices $A$ and $B$ whose product they wish to compute. The user secret shares these matrices as the shares $\hat{A}_1, \dots, \hat{A}_N$ and $\hat{B}_1, \dots, \hat{B}_N$ with an $X$-secure (or $(X, \delta)$-secure) secret sharing scheme. The user sends the matrices $\hat{A}_i$ and $\hat{B}_i$ to worker $i$ who responds with $\hat{A}_i \hat{B}_i$. Once the user has received a sufficient number of responses, they will decode the product $AB$. The \emph{recovery threshold} of an SDMM scheme is the smallest number $R$ such that $AB$ can be decoded from \emph{any} $R$ responses.

To utilize coding methods, the input matrices are partitioned to smaller pieces. One way to achieve this is using the \emph{inner partitioning} defined by
\begin{equation*}
    A = \begin{pmatrix}
        A_1 & A_2 & \cdots & A_M
    \end{pmatrix}, \quad
    B = \begin{pmatrix}
        B_1 \\ B_2 \\ \vdots \\ B_M
    \end{pmatrix}, \quad
    AB = \sum_{j=1}^M A_j B_j.
\end{equation*}
Another way is the \emph{outer partitioning} defined by
\begin{align*}
    A &= \begin{pmatrix}
        A_1 \\ A_2 \\ \vdots \\ A_K
    \end{pmatrix}, \quad B = \begin{pmatrix}
        B_1 & B_2 & \cdots & B_L
    \end{pmatrix}, \\
    AB &= \begin{pmatrix}
        A_1B_1 & A_1B_2 & \cdots & A_1B_L \\
        A_2B_1 & A_2B_2 & \cdots & A_2B_L \\
        \vdots & \vdots & \ddots & \vdots \\
        A_KB_1 & A_KB_2 & \cdots & A_KB_L
    \end{pmatrix}.
\end{align*}
There are also more general methods to partition the matrices, see \cite{karpuk2024modular, machado2022root, byrne2023straggler, aliasgari2020private}, but we will focus on the ones mentioned above.

In the linear SDMM framework proposed in \cite{makkonen2024general}, we encode the matrices using two pairs of generator matrices, $(F^\enc, F^\sec)$ and $(G^\enc, G^\sec)$, which define nested coset coding schemes as described in \cref{def:nested_coset_coding_scheme}. Therefore, we compute
\begin{align*}
    (\hat{A}_1, \dots, \hat{A}_N) &= (A_1, \dots, A_P)F^\enc + (R_1, \dots, R_X)F^\sec \\
    (\hat{S}_1, \dots, \hat{S}_N) &= (S_1, \dots, S_P)G^\enc + (S_1, \dots, S_X)G^\sec,
\end{align*}
where $R_1, \dots, R_X$ and $S_1, \dots, S_X$ are random matrices with independent entries chosen from a suitable distribution. Over the finite fields this will be the uniform distribution, but for the fields $\CC$ and $\RR$ this will be a normal distribution with a given variance $\sigma^2$. The workers compute $\hat{A}_i \hat{B}_i$ and the goal is to decode $AB$ from these results.

Instead of describing the generator matrices $(F^\enc, F^\sec)$ and $(G^\enc, G^\sec)$, we will often define (Laurent) polynomials $f$ and $g$ using the matrices as coefficients. The encoded matrices will then be defined as evaluations of these polynomials at evaluation points $\alpha_1, \dots, \alpha_N$. In particular, $\hat{A}_i = f(\alpha_i)$ and $\hat{B}_i = g(\alpha_i)$, which means that the responses are evaluations of the polynomial $h = fg$, i.e., $\hat{A}_i \hat{B}_i = f(\alpha_i) g(\alpha_i) = h(\alpha_i)$. If $f$ and $g$ are properly designed, then $AB$ can be decoded by using polynomial interpolation on the $h(\alpha_i)$. As seen from the below example, this method can equivalently be defined using the generator matrices.

\begin{example}[Construction from \cite{chang2018capacity}]
We set the security parameter as $X = 2$ and the partitioning parameters as $K = L = 2$. We partition the matrices $A$ and $B$ to $K = L$ pieces using the outer partitioning and define the polynomials
\begin{align*}
    f(z) &= A_1 + A_2 z + R_1 z^2 + R_2 z^3, \\
    g(z) &= B_1 + B_2 z^4 + S_1 z^8 + S_2 z^{12}.
\end{align*}
The shares $\hat{A}_i = f(\alpha_i)$ and $\hat{B}_i = g(\alpha_i)$ can be described by $(\hat{A}_1, \dots, \hat{A}_N) = (A_1, A_2)F^\enc + (R_1, R_2)F^\sec$ and $(\hat{B}_1, \dots, \hat{B}_N) = (B_1, B_2)G^\enc + (S_1, S_2)G^\sec$, where the generator matrices are
\begin{align*}
    F^\enc &= \begin{pmatrix}
        1 & 1 & \cdots & 1 \\
        \alpha_1 & \alpha_2 & \cdots & \alpha_N
    \end{pmatrix}, \quad
    F^\sec = \begin{pmatrix}
        \alpha_1^2 & \alpha_2^2 & \cdots & \alpha_N^2 \\
        \alpha_1^3 & \alpha_2^3 & \cdots & \alpha_N^3
    \end{pmatrix}, \\
    G^\enc &= \begin{pmatrix}
        1 & 1 & \cdots & 1 \\
        \alpha_1^4 & \alpha_2^4 & \cdots & \alpha_N^4
    \end{pmatrix}, \quad
    G^\sec = \begin{pmatrix}
        \alpha_1^8 & \alpha_2^8 & \cdots & \alpha_N^8 \\
        \alpha_1^{12} & \alpha_2^{12} & \cdots & \alpha_N^{12}
    \end{pmatrix}.
\end{align*}
The workers compute $\hat{A}_i \hat{B}_i = f(\alpha_i)g(\alpha_i) = h(\alpha_i)$, where $h = fg$. As the coefficients of $f$ and $g$ encode information about the matrices $A$ and $B$, the coefficients of $h$ encode information about the product $AB$. In particular,
\begin{align*}
    h(z) &= A_1 B_1 + A_2 B_1 z + R_1 B_1 z^2 + R_2 B_1 z^3 \\
    &+ A_1 B_2 z^4 + A_2 B_2 z^5 + R_1 B_2 z^6 + R_2 B_2 z^7 \\
    &+ A_1 S_1 z^8 + A_2 S_1 z^9 + R_1 S_1 z^{10} + R_2 S_1 z^{11} \\
    &+ A_1 S_2 z^{12} + A_2 S_2 z^{13} + R_1 S_2 z^{14} + R_2 S_2 z^{15}.
\end{align*}
The coefficients of $1, z, z^4, z^5$ encode the matrices that we need to construct the product $AB$. The coefficients can be obtained by polynomial interpolation given a sufficient number of evaluations. In this case, 16 evaluations are sufficient, since the polynomial has degree 15. We say that the above scheme has recovery threshold $R = 16$ and it can tolerate $S$ stragglers if we set $N = R + S$. This construction generalizes easily to a recovery threshold of $R = (K + X)(L + X)$. 
\end{example}

We say that a linear SDMM scheme is $(X, \delta)$-secure if both $A$ and $B$ are secret shared using $(X, \delta)$-secure secret sharing schemes, i.e., if
\begin{align*}
    &\mathcal{L}(\bm{A} \to \hat{\bm{A}}_\mathcal{X}) = \frac{1}{ts}I(\bm{A}; \hat{\bm{A}}_\mathcal{X}) \leq \delta, \quad \text{and} \\
    &\mathcal{L}(\bm{B} \to \hat{\bm{B}}_\mathcal{X}) = \frac{1}{sr} I(\bm{B}; \hat{\bm{B}}_\mathcal{X}) \leq \delta
\end{align*}
for all distributions of $\bm{A}, \bm{B}$ such that $\lVert \bm{A} \rVert_\text{max}, \lVert \bm{B} \rVert_\text{max} \leq 1$. This is essentially a restatement of the definition of $(t, \delta)$-security of secret sharing schemes.

\begin{lemma}\label{lem:SDMM_scheme_is_secure_from_secret_sharing}
A linear SDMM scheme defined by nested coset coding schemes $(F^\enc, F^\sec)$ and $(G^\enc, G^\sec)$ is $(X, \delta)$-secure if the secret sharing schemes defined by $(F^\enc, F^\sec)$ and $(G^\enc, G^\sec)$ are both $(X, \delta)$-secure.
\end{lemma}

\begin{proof}
Let $(\mu, \nu)$ denote an index to the matrices $\hat{A}_i$. Then, we can describe the encoding process as
\begin{align*}
    (\hat{A}^{\mu, \nu}_1, \hat{A}^{\mu, \nu}_2, \dots, \hat{A}^{\mu, \nu}_N) &= (A_1^{\mu, \nu}, A_2^{\mu, \nu}, \dots, A_P^{\mu, \nu})F^\enc \\
    &+ (R_1^{\mu, \nu}, R_2^{\mu, \nu}, \dots, R_X^{\mu, \nu})F^\sec.
\end{align*}
Therefore, the coordinates $A^{\mu, \nu} = (A_1^{\mu, \nu}, A_2^{\mu, \nu}, \dots, A_P^{\mu, \nu})$ are secret shared using the nested coset coding scheme $(F^\enc, F^\sec)$. As $(\mu, \nu)$ ranges over all possible indices, we will go through all entries in the matrix $A$. By \cref{lem:leakage_of_two_shares} we can bound the total leakage as
\begin{equation*}
    \mathcal{L}(\bm{A} \to \hat{\bm{A}}_\mathcal{X}) \leq \max_{\mu, \nu} \mathcal{L}(\bm{A}^{\mu, \nu} \to \hat{\bm{A}}^{\mu, \nu}_\mathcal{X}) \leq \delta.
\end{equation*}
The last inequality follows from the fact that $(F^\enc, F^\sec)$ is $(X, \delta)$-secure. A similar argument shows that $B$ is also kept secure, since $(G^\enc, G^\sec)$ is $(X, \delta)$-secure.
\end{proof}

\begin{remark}
We could define different security levels $(X_A, \delta_A)$ and $(X_B, \delta_B)$ for the matrices $A$ and $B$ if we wished, as the secret sharing of $A$ and $B$ is done independently. However, in this paper, we focus of the setting where the security parameters for $A$ and $B$ are the same.
\end{remark}

\section{SDMM over Complex Numbers}\label{sec:SDMM_over_complex_numbers}

In this section, we discuss secure distributed matrix multiplication over the complex numbers. These schemes were presented in an earlier version of this paper \cite{makkonen2022analog}. As discussed in \cref{sec:preliminaries}, polynomial interpolation is numerically stable if the evaluation points are chosen evenly on the unit circle, which is only possible when working over the complex numbers.

\subsection{Inner Partitioning}\label{sec:complex_inner_product}

The MatDot scheme and the DFT scheme for secure distributed matrix multiplication have been introduced in \cite{aliasgari2020private} and \cite{mital2022secure}, respectively, for the finite field case. We shall follow essentially the same idea to extend these constructions to the complex numbers. The matrices $A$ and $B$ are partitioned to $M$ pieces using the inner partitioning. We define the Laurent polynomials
\begin{align*}
    f(z) &= \sum_{j=1}^M A_j z^{j-1} + \sum_{j=1}^X R_j z^{M + j - 1}, \\
    g(z) &= \sum_{j=1}^M B_j z^{-(j - 1)} + \sum_{j=1}^X S_j z^j.
\end{align*}
The matrices $R_j, S_j$ for $j = 1, \dots, X$ are of appropriate size with independent entries chosen from the distribution $\CN(0, \sigma^2)$. Let $h = fg$ so that
\begin{align*}
    h(z) &= \sum_{j=1}^M \sum_{j'=1}^M A_j B_{j'} z^{j - j'} + \sum_{j=1}^M \sum_{j'=1}^X A_j S_{j'} z^{j + j' - 1} \\
    &+ \sum_{j=1}^X \sum_{j'=1}^M R_j B_{j'} z^{M + j - j'} + \sum_{j=1}^X \sum_{j'=1}^X R_j S_{j'} z^{M + j + j' - 1}.
\end{align*}
The exponents of the terms range from $-(M - 1)$ to $M + 2X - 1$, giving $2M + 2X - 1$ terms in total. Furthermore, the constant coefficient of $h$ is exactly the product $AB$.

Let $\alpha_1, \dots, \alpha_N$ be the $N$th roots of unity in $\CC$, i.e., $\alpha_i = \omega_N^i$. The workers will receive the evaluations $f(\alpha_i)$ and $g(\alpha_i)$ and will compute $h(\alpha_i) = f(\alpha_i)g(\alpha_i)$. We now present two schemes based on this encoding.

\paragraph{Complex MatDot scheme} For the MatDot scheme we set $N = 2M + 2X - 1 + S$, where $S \geq 0$ is the number of stragglers. Thus, we can interpolate the polynomial $h$ from any $R = N - S = 2M + 2X - 1$ evaluations and recover the constant coefficient $AB$. We are able to tolerate any $S$ straggling workers.

\paragraph{Complex DFT scheme} For the DFT scheme we set $N = M + 2X$. The evaluation points are $N$th roots of unity, so they have the property that
\begin{equation*}
    \sum_{i=1}^N \alpha_i^\ell = \begin{dcases*}
        N & if $N \mid \ell$ \\
        0 & otherwise
    \end{dcases*}.
\end{equation*}
The only term $h$ whose exponent is divisible by $N = M + 2X$ is the constant term, which means that $AB = \tfrac{1}{N} \sum_{i=1}^N h(\alpha_i)$. For this scheme we can not tolerate any stragglers as we need all of the responses for decoding.

\subsection{Outer Partitioning}\label{sec:complex_outer_product}

There have been many constructions for SDMM schemes using the outer partitioning, including \cite{chang2018capacity, d2020gasp, kakar2019capacity}. In this section we will introduce two simple schemes. Both use the outer partitioning where $A$ is partitioned to $K$ pieces and $B$ is partitioned to $L$ pieces. We will define two polynomials $f$ and $g$ to encode $A$ and $B$ and set $h = fg$. Let $\alpha_1, \dots, \alpha_N$ be $N$th roots of unity, i.e., $\alpha_i = \omega_N^i$. The workers receive $f(\alpha_i)$ and $g(\alpha_i)$ and compute $h(\alpha_i) = f(\alpha_i)g(\alpha_i)$.

\paragraph{Complex GASP scheme} In \cite{d2020gasp}, the following encoding was introduced as a special case of the more general GASP construction. We define the polynomials
\begin{align*}
    f(z) &= \sum_{j=1}^X R_j z^{j-1} + \sum_{j=1}^K A_j z^{K(L - 1) + X + j - 1}, \\
    g(z) &= \sum_{j=1}^X S_j z^{j-1} + \sum_{j=1}^L B_j z^{K + X - 1 + K(j - 1)}.
\end{align*}
where $R_j, S_j$ for $j = 1, \dots, X$ are of appropriate size with independent entries chosen from the distribution $\CN(0, \sigma^2)$. Let $h = fg$ so that
\begin{align*}
    h(z) &= (\text{terms with degree $\leq KL + 2X - 2$}) \\
    &+ \sum_{j=1}^K \sum_{j'=1}^L A_j B_{j'} z^{KL + 2X - 1 + j - 1 + K(j' - 1)}.
\end{align*}
The degree of $h$ is $2KL + 2X - 2$, so we need $R = 2KL + 2X - 1$ evaluations to interpolate the coefficients. The submatrices $A_j B_{j'}$ appear as the coefficients of $h$.

\paragraph{Complex A3S scheme} The following construction, known as the aligned secret sharing scheme (A3S), was introduced in \cite{kakar2019capacity}. We define the polynomials
\begin{align*}
    f(z) &= \sum_{j=1}^K A_j z^{j-1} + \sum_{j=1}^X R_j z^{K + j - 1}, \\
    g(z) &= \sum_{j=1}^L B_j z^{(K + X)(j - 1)} + \sum_{j=1}^X S_j z^{(K + X)(L - 1) + K + j - 1},
\end{align*}
where $R_j, S_j$ for $j = 1, \dots, X$ are of appropriate size with independent entries chosen from the distribution $\CN(0, \sigma^2)$. Let $h = fg$ so that
\begin{align*}
    h(z) &= \sum_{j=1}^K \sum_{j'=1}^L A_j B_{j'} z^{j - 1 + (K + X)(j' - 1)} \\
    &+ \sum_{j=1}^X \sum_{j'=1}^L R_j B_{j'} z^{K + j - 1 + (K + X)(j' - 1)} \\
    &+ (\text{terms with degree $\geq (K + X)(L - 1) + K$}).
\end{align*}
The degree of $h$ is $(K + X)(L + 1) - 2$, so we need $R = (K + X)(L + 1) - 1$ evaluations to interpolate the coefficients. The submatrices $A_j B_{j'}$ appear as the coefficients of $h$.

\subsection{Security Analysis}

From the descriptions of the schemes above, we see that the matrices $A$ and $B$ are secret shared to the workers using a Shamir secret sharing scheme over the complex numbers. In particular, the generator matrices $G^\enc$ and $G^\sec$ are generalized Vandermonde matrices with $N$th roots of unity as evaluation points. Furthermore, the exponents in $G^\sec$ are consecutive as can be seen from the polynomials $f$ and $g$ of all the proposed schemes. Therefore, we may utilize \cref{lem:SDMM_scheme_is_secure_from_secret_sharing} and \cref{thm:secret_sharing_complex_numbers}. This means that we need to choose the noise variance as
\begin{equation}\label{eq:required_noise_variance_complex}
    \sigma^2 = \frac{1}{\delta} \cdot \frac{P X^3}{4^{X - 1} \Pi(X - 1)^2} N^{2X - 2},
\end{equation}
where $\delta > 0$ is the desired level of leakage (measured in nats), $P$ is the number of partitions in the matrix ($M$, $K$ or $L$ depending on the scheme), and $N$ is the number of workers in total.

\section{SDMM over Real Numbers}\label{sec:SDMM_over_real_numbers}

In this section, we discuss how to perform SDMM when the matrices are defined over the real numbers instead of the complex numbers. Let $A \in \RR^{t \times s}$ and $B \in \RR^{s \times r}$ be matrices we want to multiply. The naive method would be to just consider the real valued matrices as a complex valued matrices and perform the methods from the previous section. However, it seems wasteful to perform complex arithmetic just to compute a product of real matrices. In this section, we propose more efficient methods to do SDMM over the real numbers.

\subsection{Encoding Real Matrices as Complex Matrices}

Let $A, B$ be real valued matrices. We will present two methods of computing the product $AB$ by encoding the matrices as complex valued matrices and computing their product. These methods were presented in \cite{ramamoorthy2022numericallya} in slightly different language. In particular, we will encode a real matrices $A, B$ into complex matrices of half the size such that the product $AB$ can be computed from these complex matrices in a natural way. We call this process \emph{complexification}\footnote{Recall that $\RR^{t \times s} \otimes_\RR \CC = \CC^{t \times s}$ is the \emph{complexification} of the real vector space $\RR^{t \times s}$. The use of the name complexification is natural in our context, since we map a real matrix to the complexification of the matrix space of half the size of the original matrix.}.

\begin{definition}[Inner complexification]\label{def:inner_complexification}

The \emph{inner complexification} of the pair $(A, B)$ of matrices with real entries is $(A', B')$, where $A' = A_1 + \imag A_2$ and $B' = B_1 - \imag B_2$, and
\begin{equation*}
    A = \begin{pmatrix}
        A_1 & A_2
    \end{pmatrix}, \quad
    B = \begin{pmatrix}
        B_1 \\ B_2
    \end{pmatrix}.
\end{equation*}
Then, $AB = A_1B_1 + A_2B_2 = \Re(A'B')$. If $A \in \RR^{t \times s}$ and $B \in \RR^{s \times r}$, then $A' \in \CC^{t \times \sfrac{s}{2}}$ and $B' \in \CC^{\sfrac{s}{2} \times r}$.
\end{definition}

\begin{definition}[Outer complexification]\label{def:outer_complexification}

The \emph{outer complexification} of the pair $(A, B)$ of matrices with real entries is $(A', B')$, where $A' = A_1 + \imag A_2$ and $B' = B_1 + \imag B_2$, and
\begin{equation*}
    A = \begin{pmatrix}
        A_1 \\ A_2
    \end{pmatrix}, \quad
    B = \begin{pmatrix}
        B_1 & B_2
    \end{pmatrix}.
\end{equation*}
Then, the product $AB$ can be written as
\begin{equation}\label{eq:outer_complexification_assembly}
    AB =
    \tfrac{1}{2} \begin{pmatrix}
        \Re(A' B') + \Re(A' \overline{B'}) & \Im(A' B') - \Im(A' \overline{B'}) \\
        \Im(A' B') + \Im(A' \overline{B'}) & -\Re(A' B') + \Re(A' \overline{B'})
    \end{pmatrix},
\end{equation}
i.e., $AB$ can be computed from $A'B'$ and $A'\overline{B'}$. If $A \in \RR^{t \times s}$ and $B \in \RR^{s \times r}$, then $A' \in \CC^{\sfrac{t}{2} \times s}$ and $B' \in \CC^{s \times \sfrac{r}{2}}$.
\end{definition}

The authors of \cite{ramamoorthy2022numericallya} considered a similar concept, which they call the rotation matrix embedding, but their language and notation significantly differs from the above. In both the inner and outer complexification methods, the matrix $A'$ is a complex matrix of half the size of $A$, so it takes the same number of real parameters to describe. The same holds for $B$ and $B'$. To compute the matrix product $AB$ we need either $\Re(A'B')$ or $(A'B', A' \overline{B'})$ depending on the method. Below we will design SDMM schemes specifically to compute these products efficiently.

\subsection{Inner Partitioning}

Let $(A', B')$ be the inner complexification of $(A, B)$, i.e., $AB = \Re(A'B')$. In this section, we propose a method to retrieve just the real part of $A' B'$ instead of the whole complex matrix. This comes at the cost of increasing the number of workers slightly.

We will encode $A' \in \CC^{t \times \sfrac{s}{2}}$ and $B' \in \CC^{\sfrac{s}{2} \times r}$ using the method presented in \cref{sec:complex_inner_product}, i.e., by partitioning them using the inner partitioning to $M$ pieces and by creating the Laurent polynomials
\begin{equation}\label{eq:real_inner_polynomials}
\begin{split}
    f(z) &= \sum_{j=1}^M A'_j z^{j-1} + \sum_{j=1}^X R'_j z^{M + j - 1}, \\
    g(z) &= \sum_{j=1}^M B'_j z^{-(j - 1)} + \sum_{j=1}^X S'_j z^j.
\end{split}
\end{equation}
The matrices $R'_j$ and $S'_j$ are of appropriate size with independent entries chosen from the distribution $\CN(0, \sigma^2)$. The constant term of $fg$ is equal to $A'B'$. As before, the workers receive evaluations of these polynomials at the evaluation points $\alpha_i = \omega_N^i$ for $i \in [N]$, i.e., $\hat{A}_i = f(\alpha_i) \in \CC^{t \times \sfrac{s}{2M}}$ and $\hat{B}_i = g(\alpha_i) \in \CC^{\sfrac{s}{2M} \times r}$. However, instead of computing $\hat{A}_i \hat{B}_i$, the workers compute $\Re(\hat{A}_i \hat{B}_i) \in \RR^{t \times r}$. By noticing that $\overline{\alpha_i} = \alpha_i^{-1}$ for all $i \in [N]$, we can write the above as
\begin{align*}
    \Re(f(\alpha_i)g(\alpha_i)) &= \tfrac{1}{2}(f(\alpha_i)g(\alpha_i) + \overline{f(\alpha_i)g(\alpha_i)}) \\
    &= \tfrac{1}{2}(f(\alpha_i)g(\alpha_i) + \overline{f}(\alpha_i^{-1})\overline{g}(\alpha_i^{-1})).
\end{align*}
Therefore, the responses are evaluations of the Laurent polynomial
\begin{equation*}\label{eq:real_inner_h}
    h(z) = \tfrac{1}{2}(f(z)g(z) + \overline{f}(z^{-1})\overline{g}(z^{-1})).
\end{equation*}
The constant term of $h$ is $\frac{1}{2}(A' B' + \overline{A' B'}) = \Re(A'B') = AB$. This polynomial has terms with exponents from $-(M + 2X - 1)$ to $M + 2X - 1$. We have described the worker computation as computing the real part of a complex product, but we could also describe it as the product of real matrices, since
\begin{equation}\label{eq:inner_complexification_worker_computation}
     \Re(\hat{A}_i \hat{B}_i)
     = \underbrace{\begin{pmatrix}
         \Re(\hat{A}_i) & \Im(\hat{A}_i)
     \end{pmatrix}}_{\in \RR^{t \times \sfrac{s}{M}}} \cdot \underbrace{\begin{pmatrix}
         \Re(\hat{B}_i) \\ -\Im(\hat{B}_i)
     \end{pmatrix}}_{\in \RR^{\sfrac{s}{M} \times r}}.
\end{equation}
This encoding method can be used to design the following two SDMM schemes over the real numbers.

\paragraph{Real MatDot scheme} The polynomial $h$ has at most $R = 2M + 4X - 1$ terms, so it can be interpolated from any $R$ evaluations. Thus, we can have $N = R + S$ workers and tolerate $S$ straggling workers. The interpolation is numerically stable since the evaluation points are $N$th roots of unity. Notice that the recovery threshold of $R = 2M + 4X - 1$ is higher than the analogous recovery threshold of $2M + 2X - 1$ for the complex numbers.

\paragraph{Real DFT scheme} Let $N = M + 2X$. As all the exponents in $h$ are between $-(N - 1)$ and $N - 1$, the only term whose degree is divisible by $N$ is the constant term. Therefore, by the properties of roots of units,
\begin{equation*}
    \frac{1}{N} \sum_{i=1}^N h(\alpha_i) = \Re(A'B') = AB.
\end{equation*}

The SDMM procedures over the real numbers using the inner partitioning are written in \cref{alg:real_inner_scheme}.

\begin{algorithm}[t]
    \caption{The real inner partitioning scheme}
    \label{alg:real_inner_scheme}
    \begin{algorithmic}[1]
        \Procedure{Real Matrix Product}{$A, B$}
            \LComment{$N$ workers, partitioning $M$, security $X$, recovery threshold $R$, noise variance $\sigma^2$}
            \State $(A', B') \gets$ \Call{Inner Complexification}{$A, B$}
            \State $\{A'_j, B'_j\}_{j \in [M]} \gets$ \Call{Inner Partition}{$A', B', M$}
            \State $\{R'_j, S'_j\}_{j \in [X]} \gets$ random entries from $\CN(0, \sigma^2)$
            \For{worker $i \in [N]$}
            \State $\alpha_i \gets \exp(2\imag \pi i / N)$
            \State $\hat{A}'_i, \hat{B}'_i \gets f(\alpha_i), g(\alpha_i)$ \Comment{$f, g$ defined in \eqref{eq:real_inner_polynomials}}
            \State $h(\alpha_i) \gets \Re(\hat{A}'_i \hat{B}'_i)$ \Comment{$h$ defined in \eqref{eq:real_inner_h}}
            \EndFor
            \State Compute the constant coefficient $h_0$ by interpolation (MatDot) or averaging (DFT)
            \State $AB \gets h_0$
            \State \textbf{return} $AB$
        \EndProcedure
    \end{algorithmic}
\end{algorithm}

\subsection{Outer Partitioning}

Let $(A', B')$ be the outer complexification of $(A, B)$, i.e., $AB$ can be computed from $A'B'$ and $A' \overline{B'}$. In this section, we propose methods to encode $A'$ and $B'$ such that both $A' B'$ and $A' \overline{B'}$ can be decoded from the responses. This will again come at a small cost of an increased number of workers.

We will partition the matrices $A'$ and $B'$ to $K$ and $L$ pieces using the outer partitioning. We define the polynomials $f$ and $g$ to encode $A'$ and $B'$. The workers receive $\hat{A}'_i = f(\alpha_i) \in \CC^{\sfrac{t}{2K} \times s}$ and $\hat{B}'_i = g(\alpha_i) \in \CC^{s \times \sfrac{r}{2L}}$, where $\alpha_i$ is an $N$th root of unity. The workers compute
\begin{equation*}
    \hat{A}'_i \hat{B}'_i = f(\alpha_i)g(\alpha_i)~ \text{and} ~
    \hat{A}'_i \overline{\hat{B}'_i} = f(\alpha_i)\overline{g(\alpha_i)} = f(\alpha_i)\overline{g}(\alpha_i^{-1})
\end{equation*}
as $\overline{\alpha_i} = \alpha_i^{-1}$ for all evaluation points. Therefore, the workers compute evaluations of the Laurent polynomials $h^{(+)}(z) = f(z)g(z)$ and $h^{(-)}(z) = f(z) \overline{g}(z^{-1})$. Even though we have described the computation as two complex matrix multiplications, we could have described it as the product of real matrices, since
\begin{equation}\label{eq:outer_complexification_worker_computation}
    \underbrace{\begin{pmatrix}
        \Re(\hat{A}'_i) \\
        \Im(\hat{A}'_i)
    \end{pmatrix}}_{\in \RR^{\sfrac{t}{K} \times s}} \underbrace{\begin{pmatrix}
        \Re(\hat{B}'_i) & \Im(\hat{B}'_i)
    \end{pmatrix}}_{\in \RR^{s \times \sfrac{r}{L}}} = \begin{pmatrix}
        \Re(\hat{A}'_i) \Re(\hat{B}'_i) & \Re(\hat{A}'_i)\Im(\hat{B}'_i) \\
        \Im(\hat{A}'_i) \Re(\hat{B}'_i) & \Im(\hat{A}'_i)\Im(\hat{B}'_i)
    \end{pmatrix}
\end{equation}
and
\begin{align*}
    \hat{A}'_i \hat{B}'_i &= \left( \Re(\hat{A}'_i) \Re(\hat{B}'_i) - \Im(\hat{A}'_i)\Im(\hat{B}'_i) \right)
    + \imag \left( \Re(\hat{A}'_i)\Im(\hat{B}'_i) + \Im(\hat{A}'_i) \Re(\hat{B}'_i) \right) \\
    \hat{A}'_i \overline{\hat{B}'_i} &= \left( \Re(\hat{A}'_i) \Re(\hat{B}'_i) + \Im(\hat{A}'_i)\Im(\hat{B}'_i) \right)
    + \imag \left( -\Re(\hat{A}'_i)\Im(\hat{B}'_i) + \Im(\hat{A}'_i) \Re(\hat{B}'_i) \right).
\end{align*}
We present the following two SDMM schemes over the real numbers.

\paragraph{Real GASP scheme} Encode $A'$ and $B'$ as the polynomials
\begin{equation}\label{eq:real_gasp_polynomials}
\begin{split}
    f(z) &= \sum_{j=1}^X R'_j z^{j - 1} + \sum_{j=1}^K A'_j z^{KL + 2X - 1 + j - 1}, \\
    g(z) &= \sum_{j=1}^X S'_j z^{j - 1} + \sum_{j=1}^L B'_j z^{K + X - 1 + K(j - 1)}.
\end{split}
\end{equation}
The matrices $R'_j$ and $S'_j$ are of appropriate size with independent entries chosen from the distribution $\CN(0, \sigma^2)$. Notice that the random matrices are encoded in the terms with low degree. We can write
\begin{equation}\label{eq:real_gasp_h_plus}
\begin{split}
    h^{(+)}(z) &= (\text{terms with degree $\leq KL + K + 3X - 3$}) \\
    &+ \sum_{j=1}^K \sum_{j'=1}^L A'_j B'_{j'} z^{KL + K + 3X - 2 + j - 1 + K(j' - 1)}.
\end{split}
\end{equation}
The submatrices $A'_j B'_{j'}$ appear as the coefficients of $h^{(+)}$ and the degree of $h^{(+)}$ is $2KL + K + 3X - 3$. On the other hand,
\begin{equation}\label{eq:real_gasp_h_minus}
\begin{split}
    h^{(-)}(z) &= (\text{terms with degree $\leq X - 1$ or $\geq KL + X$}) \\
    &+ \sum_{j=1}^K \sum_{j'=1}^L A'_j \overline{B'_{j'}} z^{K(L - 1) + X + j - 1 - K(j' - 1)}.
\end{split}
\end{equation}
Again, the submatrices $A'_j \overline{B'_j}$ appear as the coefficients of $h^{(-)}$. The terms in $h^{(-)}$ have degrees ranging from $-(KL + X - 1)$ to $KL + K + 2X - 2$ for a total of $2KL + K + 3X - 2$ terms. Therefore, we can interpolate $h^{(+)}$ and $h^{(-)}$ from any $R = 2KL + K + 3X - 2$ evaluations. This allows us to compute $A' B'$ and $A' \overline{B}'$.

\paragraph{Real A3S scheme} Encode $A'$ and $B'$ as the polynomials
\begin{equation}\label{eq:real_a3s_polynomials}
\begin{split}
    f(z) &= \sum_{j=1}^K A'_j z^{j - 1} + \sum_{j=1}^X R'_j z^{K + j - 1}, \\
    g(z) &= \sum_{j=1}^L B'_j z^{(K + X)(j - 1)} + \sum_{j=1}^X S'_j z^{(K + X)L + j - 1}.
\end{split}
\end{equation}
The matrices $R'_j$ and $S'_j$ are of appropriate size with independent entries chosen from the distribution $\CN(0, \sigma^2)$. The polynomial $h^{(+)}$ can be written as
\begin{equation}\label{eq:real_a3s_h_plus}
\begin{split}
    h^{(+)}(z) &= \sum_{j=1}^K \sum_{j'=1}^L A'_j B'_{j'} z^{j - 1 + (K + X)(j' - 1)} \\
    &+ \sum_{j=1}^X \sum_{j'=1}^L R'_j B'_{j'} z^{K + j - 1 + (K + X)(j' - 1)} \\
    &+ (\text{terms of degree $\geq (K + X)L$}).
\end{split}
\end{equation}
The submatrices $A'_j B'_{j'}$ appear as coefficients of $h^{(+)}$. The degree of $h^{(+)}$ is $(K + X)(L + 1) + X - 2$, so $(K + X)(L + 1) + X - 1$ evaluations are needed to interpolate the polynomial. On the other hand,
\begin{equation}\label{eq:real_a3s_h_minus}
\begin{split}
    h^{(-)}(z) &= \sum_{j=1}^K \sum_{j'=1}^L A'_j \overline{B'_{j'}} z^{j - 1 - (K + X)(j' - 1)} \\
    &+ \sum_{j=1}^X \sum_{j'=1}^L R'_j \overline{B'_{j'}} z^{K + j - 1 - (K + X)(j' - 1)} \\
    &+ (\text{terms with degree $\leq -(K + X)(L - 1) - 1$}).
\end{split}
\end{equation}
The submatrices $A'_j \overline{B'_{j'}}$ appear as coefficients of $h^{(-)}$. The terms in $h^{(-)}$ range from $-(K + X)L - X + 1$ to $K + X - 1$ for a total of $(K + X)(L + 1) + X - 1$ terms. Therefore, we need $R = (K + X)(L + 1) + X - 1$ evaluations of both polynomials to be able to decode all the required matrix products.

The SDMM procedures over the real numbers using the outer partitioning are written in \cref{alg:real_outer_scheme}.

\begin{algorithm}[t]
    \caption{The real outer partitioning scheme}
    \label{alg:real_outer_scheme}
    \begin{algorithmic}[1]
        \Procedure{Real Matrix Product}{$A, B$}
            \LComment{$N$ workers, partitioning $K, L$, security $X$, recovery threshold $R$, noise variance $\sigma^2$}
            \State $(A', B') \gets$ \Call{Outer Complexification}{$A, B$}
            \State $\{A'_j, B'_{j'}\}_{\substack{j \in [K]\\ j' \in [L]}} \gets$ \Call{Outer Partition}{$A', B', K, L$}
            \State $\{R'_j, S'_j\}_{j \in [X]} \gets$ random entries from $\CN(0, \sigma^2)$
            \For{worker $i \in [N]$}
            \State $\alpha_i \gets \exp(2\imag \pi i / N)$
            \State $\hat{A}'_i, \hat{B}'_i \gets f(\alpha_i), g(\alpha_i)$ \Comment{$f, g$ defined in \eqref{eq:real_gasp_polynomials} or \eqref{eq:real_a3s_polynomials}}
            \State $h^{(+)}(\alpha_i), h^{(-)}(\alpha_i) \gets \hat{A}'_i \hat{B}'_i,  \hat{A}'_i \overline{\hat{B}'_i}$ \Comment{$h^{(+)}, h^{(-)}$ defined in \eqref{eq:real_gasp_h_plus}, \eqref{eq:real_gasp_h_minus} or \eqref{eq:real_a3s_h_plus}, \eqref{eq:real_a3s_h_minus}}
            \EndFor
            \State Interpolate $h^{(+)}$ and $h^{(-)}$ from any $R$ results
            \State Assemble $A' B'$ and $A' \overline{B}'$ from the coefficients of $h^{(+)}$ and $h^{(-)}$
            \State Assemble $AB$ from $A' B'$ and $A' \overline{B'}$ according to \eqref{eq:outer_complexification_assembly}
            \State \textbf{return} $AB$
        \EndProcedure
    \end{algorithmic}
\end{algorithm}

\subsection{Security Analysis}

We assume that the entries in $A$ and $B$ have absolute value at most 1. Then, the modulus of the elements in $A'$ and $B'$ is at most $\sqrt{2}$. The encoding process in the proposed schemes is a a Shamir secret sharing scheme over $\CC$ with generalized Vandermonde matrices $(F^\enc, F^\sec)$ and $(G^\enc, G^\sec)$, where $F^\sec$ and $G^\sec$ have consecutive exponents. Therefore, the noise variance can be set at
\begin{equation*}
    \sigma^2 = \frac{1}{\delta} \cdot \frac{2 P X^3}{4^{X - 1} \Pi(X - 1)^2} N^{2X - 2}
\end{equation*}
according to \cref{lem:SDMM_scheme_is_secure_from_secret_sharing} and \cref{thm:secret_sharing_complex_numbers}, where $\delta > 0$ is the desired level of leakage (measured in nats), $P$ is the number of partitions in the matrix ($M, K$ or $L$ depending on the scheme), and $N$ is the number of workers in total. The factor $2$ is set to compensate for the increased modulus.

\subsection{Complexity Analysis}

In this section, we will compute the computational complexity of the proposed SDMM schemes over the real numbers utilizing the inner and outer complexifications. In particular, we will compare the number of real operations required in different phases of the algorithms to a generic linear SDMM scheme over $\RR$ and $\CC$ with the same partitioning parameters and number of workers. We compare to generic schemes to showcase that using the complexification does not significantly increase the computational cost. We would like to remind that a scheme defined over $\RR$ would suffer from numerical stability issues, at least if it were based on polynomial interpolation. Notice that using the complexification does not change the size of the matrices sent to the worker compared to a linear SDMM scheme over $\RR$ with the same partitioning parameters, see \eqref{eq:inner_complexification_worker_computation} and \eqref{eq:outer_complexification_worker_computation}. Hence, the communication costs and the worker computational cost is not increased by utilizing the complexification.

The encoding and decoding in SDMM mostly consists of computing linear combinations of (large) matrices. Computing a linear combination of $K$ matrices in $\RR^{t \times s}$ with real coefficients requires $(2K - 1) \cdot ts$ real operations (additions and multiplications). On the other hand, computing a linear combination of $K$ matrices in $\CC^{t \times s}$ with complex coefficients requires $(8K - 2) \cdot ts$ real operations, since each complex multiplication takes 6 real operations and each complex addition takes 2 real operations. We will not consider the computations required to invert the interpolation matrices since the complexity does not scale with the sizes of the input matrices $A$ and $B$.

The computational costs are compiled in \cref{tab:computational_costs}. We count the number of additions and multiplications, since their cost in hardware is roughly equal on modern computer architectures. Even though computing linear combinations of complex matrices is about four times as expensive as computing a linear combination of real matrices, the complex matrices in the complexified schemes are half the size due to the complexification. Therefore, the encoding and decoding costs about twice as much compared to a linear SDMM scheme with the same partitioning. The amount of random symbols from a normal distribution is the same for the complexification based schemes compared to regular linear SDMM schemes over the real numbers.

\begin{table*}[t]
    \centering
    \begin{tabular}{|l|c|c|c|}
        \hline
        Method & Encoding $A \in \RR^{t \times s}$ & Encoding $B \in \RR^{s \times r}$ & Decoding \\ \hline
        Linear SDMM over $\RR$ (inner partitioning) & $(2(M + X) - 1) \cdot \frac{ts}{M}$ & $(2(M + X) - 1) \cdot \frac{sr}{M}$ & $(2R - 1) \cdot tr$ \\
        \cref{alg:real_inner_scheme} (inner partitioning) & $(4(M + X) - 1) \cdot \frac{ts}{M}$ & $(4(M + X) - 1) \cdot \frac{sr}{M}$ & $(2R - 1) \cdot tr$ \\
        Linear SDMM over $\CC$ (inner partitioning) & $(8(M + X) - 2) \cdot \frac{ts}{M}$ & $(8(M + X) - 2) \cdot \frac{sr}{M}$ & $(8R - 2) \cdot tr$ \\ \hline
        Linear SDMM over $\RR$ (outer partitioning) & $(2(K + X) - 1) \cdot \frac{ts}{K}$ & $(2(K + X) - 1) \cdot \frac{sr}{K}$ & $(2R - 1) \cdot tr$ \\
        \cref{alg:real_outer_scheme} (outer partitioning) & $(4(K + X) - 1) \cdot \frac{ts}{K}$ & $(4(L + X) - 1) \cdot \frac{sr}{L}$ & $(4R - 1) \cdot tr$ \\
        Linear SDMM over $\CC$ (outer partitioning) & $(8(K + X) - 2) \cdot \frac{ts}{K}$ & $(8(L + X) - 2) \cdot \frac{sr}{L}$ & $(8R - 2) \cdot tr$ \\ \hline
    \end{tabular}
    \caption{The number of real number operations to perform the encoding and decoding of SDMM for real valued matrices. The scheme over $\CC$ utilizes embedding the matrices in the complex numbers yielding a high cost.}
    \label{tab:computational_costs}
\end{table*}

\section{Numerical Experiments}\label{sec:numerical_experiments}

\pgfplotsset{
    tick label style={
        font=\footnotesize,
        /pgf/number format/sci,
    },
    xlabel style={
        font=\footnotesize
    },
    ylabel style={
        font=\footnotesize
    }
} 

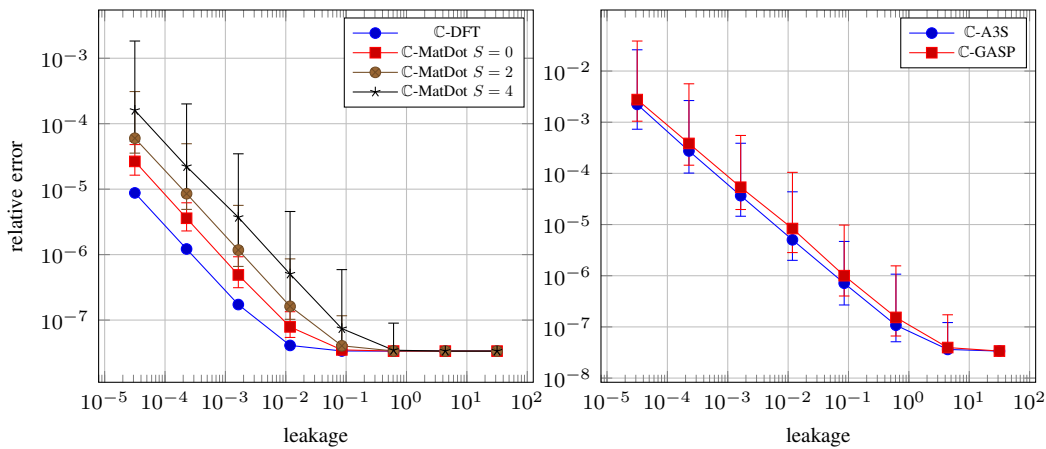
\begin{figure}[t]
    \begin{center}
    \begin{tikzpicture}
    \begin{groupplot}[
        group style={group size=2 by 1, horizontal sep=0.9cm},
        scale only axis,
        width=0.7\linewidth,
        height=0.3\linewidth,
        ymode=log,
        xmode=log,
        grid=major,
        ytickten={-2, -3, ..., -8},
        enlarge y limits=auto,
        legend style={nodes={scale=0.6, transform shape}},
        ]
    \nextgroupplot[width=0.35\linewidth, xlabel={leakage}, ylabel={relative error}]
        \addplot+ [
            error bars/.cd,
                y dir=both, y explicit,
        ] table [
            y error plus=ey+,
            y error minus=ey-,
        ] {data/complex_dft_errors_leakage.dat};
        \addplot+ [
            error bars/.cd,
                y dir=both, y explicit,
        ] table [
            y error plus=ey+,
            y error minus=ey-,
        ] {data/complex_matdot_errors_leakage_0.dat};
        \addplot+ [
            error bars/.cd,
                y dir=both, y explicit,
        ] table [
            y error plus=ey+,
            y error minus=ey-,
        ] {data/complex_matdot_errors_leakage_2.dat};
        \addplot+ [
            error bars/.cd,
                y dir=both, y explicit,
        ] table [
            y error plus=ey+,
            y error minus=ey-,
        ] {data/complex_matdot_errors_leakage_4.dat};
        \legend{$\CC$-DFT, $\CC$-MatDot $S = 0$, $\CC$-MatDot $S = 2$, $\CC$-MatDot $S = 4$}
    \nextgroupplot[width=0.35\linewidth, xlabel={leakage}]
        \addplot+ [
            error bars/.cd,
                y dir=both, y explicit,
        ] table [
            y error plus=ey+,
            y error minus=ey-,
        ] {data/complex_a3s_errors_leakage.dat};
        \addplot+ [
            error bars/.cd,
                y dir=both, y explicit,
        ] table [
            y error plus=ey+,
            y error minus=ey-,
        ] {data/complex_gasp_errors_leakage.dat};
        \legend{$\CC$-A3S, $\CC$-GASP}
    \end{groupplot}
    \end{tikzpicture}
    \end{center}
    \caption{The relative Frobenius error of the proposed complex SDMM schemes as a function of the leakage. The partitioning used is $M = 8$ for the left plot and $K = L = 4$. The security parameter is $X = 3$. The number of stragglers is $S = 4$ on the right. The bars represent 90\% confidence intervals around the median.}
    \label{fig:complex_inner_outer_leakege_error}
\end{figure}

\begin{figure}[t]
    \begin{center}
    \begin{tikzpicture}
    \begin{groupplot}[
        group style={group size=2 by 1},
        scale only axis,
        width=0.7\linewidth,
        height=0.3\linewidth,
        ymode=log,
        xmode=log,
        grid=major,
        ytickten={0, -1, ..., -4},
        enlarge y limits=auto,
        legend style={nodes={scale=0.6, transform shape}},
        ] 
    \nextgroupplot[width=0.35\linewidth, xlabel={leakage}, ylabel={relative error}]
        \addplot+ [
            error bars/.cd,
                y dir=both, y explicit,
        ] table [
            y error plus=ey+,
            y error minus=ey-,
        ] {data/real_dft_errors_leakage.dat};
        \addplot+ [
            error bars/.cd,
                y dir=both, y explicit,
        ] table [
            y error plus=ey+,
            y error minus=ey-,
        ] {data/real_matdot_errors_leakage_0.dat};
        \addplot+ [
            error bars/.cd,
                y dir=both, y explicit,
        ] table [
            y error plus=ey+,
            y error minus=ey-,
        ] {data/real_matdot_errors_leakage_2.dat};
        \addplot+ [
            error bars/.cd,
                y dir=both, y explicit,
        ] table [
            y error plus=ey+,
            y error minus=ey-,
        ] {data/real_matdot_errors_leakage_4.dat};
        \legend{$\RR$-DFT, $\RR$-MatDot $S = 0$, $\RR$-MatDot $S = 2$, $\RR$-MatDot $S = 4$}
        
    \nextgroupplot[width=0.35\linewidth, xlabel={leakage}]
        \addplot+ [
            error bars/.cd,
                y dir=both, y explicit,
        ] table [
            y error plus=ey+,
            y error minus=ey-,
        ] {data/real_a3s_errors_leakage.dat};
        \addplot+ [
            error bars/.cd,
                y dir=both, y explicit,
        ] table [
            y error plus=ey+,
            y error minus=ey-,
        ] {data/real_gasp_errors_leakage.dat};
        \legend{$\RR$-A3S, $\RR$-GASP}
    \end{groupplot}
    \end{tikzpicture}
    \end{center}
    \caption{The relative Frobenius error of the proposed real SDMM schemes as a function of the leakage. The partitioning used is $M = 8$ for the left plot and $K = L = 4$. The security parameter is $X = 3$. The number of stragglers is $S = 4$ on the right. The bars represent 90\% confidence intervals around the median.}
    \label{fig:real_inner_outer_leakege_error}
\end{figure}
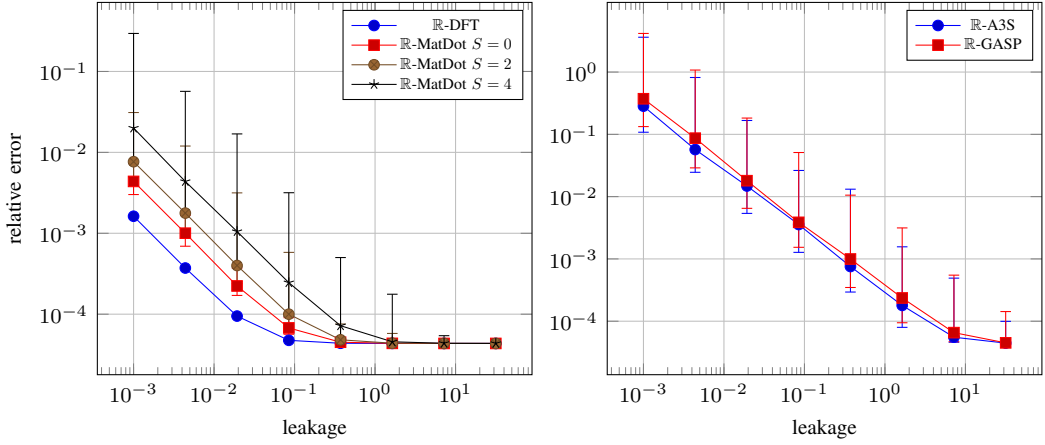

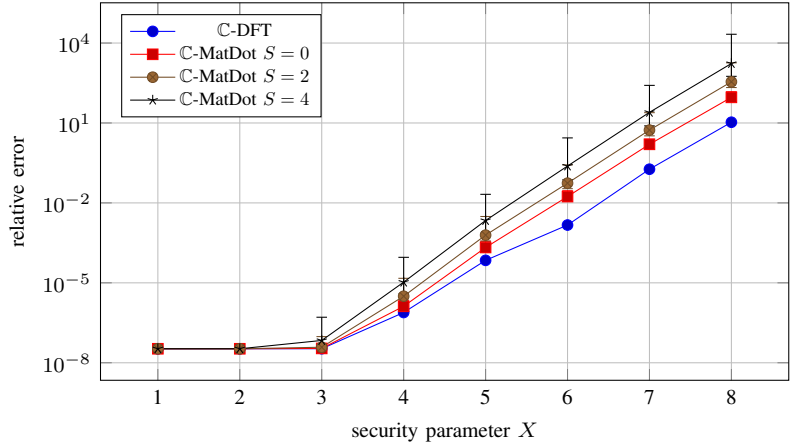
\begin{figure}[t]
    \begin{center}
    \begin{tikzpicture}
    \begin{semilogyaxis}[
        width=0.65\linewidth,
        height=0.4\linewidth,
        grid=major,
        legend style={nodes={scale=0.7, transform shape}, legend pos=north west},
        xlabel={security parameter $X$},
        ylabel={relative error},
        xtick={1, 2, 3, 4, 5, 6, 7, 8},
        xticklabels={1, 2, 3, 4, 5, 6, 7, 8},
        ytickten={4, 1, ..., -8},
        scaled x ticks=true,
        ]
        \addplot+ [
            error bars/.cd,
                y dir=both, y explicit,
        ] table [
            y error plus=ey+,
            y error minus=ey-,
        ] {data/complex_dft_errors_security.dat};
        \addplot+ [
            error bars/.cd,
                y dir=both, y explicit,
        ] table [
            y error plus=ey+,
            y error minus=ey-,
        ] {data/complex_matdot_errors_security_0.dat};
        \addplot+ [
            error bars/.cd,
                y dir=both, y explicit,
        ] table [
            y error plus=ey+,
            y error minus=ey-,
        ] {data/complex_matdot_errors_security_2.dat};
        \addplot+ [
            error bars/.cd,
                y dir=both, y explicit,
        ] table [
            y error plus=ey+,
            y error minus=ey-,
        ] {data/complex_matdot_errors_security_4.dat};
        \legend{$\CC$-DFT, $\CC$-MatDot $S = 0$, $\CC$-MatDot $S = 2$, $\CC$-MatDot $S = 4$}
    \end{semilogyaxis}
    \end{tikzpicture}
    \end{center}
    \caption{The relative Frobenius error as a function of the security parameter $X$ for inner partitioning schemes over the complex numbers. The bars represent 90\% confidence intervals around the median.}
    \label{fig:complex_inner_security_error}
\end{figure}

To experimentally test the proposed schemes, we choose random matrices $A$ and $B$ uniformly from either the interval $[-1, 1]$ or from the unit disk in $\CC$. We then run the algorithms with different partitioning parameters $M, K, L$, security parameters $X$, and leakage levels $\delta$, and measure the error in the result. The true result is computed as the product of the matrices without doing SDMM and the error is measured as the relative error between the computed product and the true product. We perform all of the computations with single-precision real or complex floating point numbers. Thus, every real number is represented as 4 bytes, and every complex number is represented as 8 bytes. The input matrices $A$ and $B$ were chosen to be $256 \times 256$ matrices.

The results of the numerical experiments can be seen in \cref{fig:complex_inner_outer_leakege_error,fig:real_inner_outer_leakege_error,fig:complex_inner_security_error}. It is clear from the figures that there is a trade-off between the accuracy and the leakage of the schemes, i.e., a lower leakage generally implies a higher relative error. From \cref{fig:complex_inner_outer_leakege_error,fig:real_inner_outer_leakege_error} we see that the schemes over the real numbers generally have a larger error compared to the schemes over the complex numbers due to the higher number of worker nodes.

In \cref{fig:complex_inner_security_error} we see that the relative error in the result grows exponentially as a function of the security parameter $X$. This comes from the fact that the required noise variance grows exponentially in $X$ as seen in \eqref{eq:required_noise_variance_complex}. Therefore, the proposed schemes are not well suited for a large security parameter $X$.

We do not compare our proposed schemes to ones based purely on real number arithmetic since it is not clear how such a scheme would be made secure, i.e., what the required noise variance would be. It is reasonable to assume that such a scheme would have a large condition number if it were based on polynomial interpolation, since Vadermonde matrices are ill-conditioned over the real numbers. On the other hand, our constructions achieve a polynomial condition number in $N$ if the number of stragglers $S = N - R$ is kept constant, see \cref{thm:large_vandermonde_inverse_2_norm_upper_bound}.

\section{Conclusion}

We constructed schemes for SDMM over the complex and real numbers that have small leakage and good stability. For the complex numbers we utilized polynomial interpolation with roots of unity, which provides good numerical stability. Over the real numbers we proposed a method called the complexification, which can be used to embed the real computations into the complex numbers in a more efficient way than the trivial embedding.

\appendix\label{sec:appendix}

In this appendix, we consider $n \times k$ Vandermonde matrices to be of the following form
\begin{equation*}
    V = \begin{pmatrix}
        1 & \alpha_1 & \alpha_1^2 & \cdots & \alpha_1^{k - 1} \\
        1 & \alpha_2 & \alpha_2^2 & \cdots & \alpha_2^{k - 1} \\
        1 & \alpha_3 & \alpha_3^2 & \cdots & \alpha_3^{k - 1} \\
        \vdots & \vdots & \vdots & \ddots & \vdots \\
        1 & \alpha_n & \alpha_n^2 & \cdots & \alpha_n^{k - 1}
    \end{pmatrix},
\end{equation*}
where the elements $\alpha_1, \dots, \alpha_n$ are the evaluation points of the Vandermonde matrix. In particular, $V_{ij} = \alpha_i^{j - 1}$ for $i \in [n], j \in [k]$. An $m \times m$ Vandermonde matrix is invertible if and only if the evaluation points are pairwise distinct. It is well known that Vandermonde matrices over $\RR$ are ill-conditioned \cite{pan2016how}. On the other hand, Vandermonde matrices over $\CC$ with evaluation points being equally spaced on the unit circle are perfectly conditioned. In this section, we consider $m \times m$ Vandermonde matrices whose evaluation points are distinct $n$th roots of unity for $m \leq n$. Let $\omega_n = \exp(\tfrac{2\pi \imag}{n})$ be a primitive $n$th root of unity for $n \geq 1$.

Let $V$ be an $n \times n$ Vandermonde matrix whose evaluation points are the distinct $n$th roots of unity, i.e., $\alpha_i = \omega_n^{i - 1}$. Then, $V^{-1} = \tfrac{1}{n} V^*$. This means that $\tfrac{1}{\sqrt{n}} V$ is unitary and $\lVert \tfrac{1}{\sqrt{n}} V \rVert_2 = 1$, so $\lVert V \rVert_2 = \lVert V^* \rVert_2 = \sqrt{n}$. Therefore, $\kappa_2(V) = \lVert V \rVert_2 \lVert V^{-1} \rVert_2 = 1$.

Let $W$ be an $m \times m$ Vandermonde matrix whose evaluation points $\alpha_i$ are distinct $n$th roots of unity for $m \leq n$. We will denote $d = n - m$. Let us write $\alpha_i = \omega_n^{k_i}$ for integers $0 \leq k_1 < \cdots < k_m < n$. Additionally, let us denote $\mathcal{K} = \{k_1, \dots, k_m\}$ and $\mathcal{L} = \{0, 1, \dots, n - 1 \} \setminus \mathcal{K}$. We start by stating an upper bound on the norm of $W$.

\begin{lemma}\label{lem:vandermonde_2_norm_upper_bound}
Let $W$ be an $m \times m$ Vandermonde matrix whose evaluation points are distinct $n$th roots of unity for $m \leq n$. Then, $\lVert W \rVert_2 \leq \sqrt{n}$.
\end{lemma}

\begin{proof}
Let $x \in \CC^m$ be a vector. Add rows to $W$ to obtain an $n \times m$ Vandermonde matrix $\overline{W}$ whose evaluation points are the $n$th roots of unity. By further adding columns to $\overline{W}$ we obtain the $n \times n$ Vandermonde matrix $V$ whose evaluation points are the $n$th roots of unity. Then,
\begin{equation*}
    \lVert Wx \rVert_2 \leq \lVert \overline{W}x \rVert_2 = \lVert V \overline{x} \rVert_2,
\end{equation*}
where $\overline{x} = (x, 0) \in \CC^n$. We have that $\lVert x \rVert_2 = \lVert \overline{x} \rVert_2$, so
\begin{equation*}
    \lVert W \rVert_2 = \max_{\lVert x \rVert_2 = 1} \lVert Wx \rVert_2 \leq \max_{\lVert \overline{x} \rVert_2 = 1} \lVert V \overline{x} \rVert_2 = \lVert V \rVert_2 = \sqrt{n}. \qedhere
\end{equation*}
\end{proof}

To compute the condition number we also need to know the norm of the inverse of $W$. We start doing this by stating the following decomposition for our matrix.

\begin{lemma}\label{lem:vandermonde_inverse_decomposition}
The inverse of $W$ can be written as $W^{-1} = B^{-1} C A$, where $B$ is an $m \times m$ Vandermonde matrix with evaluation points $\beta_1, \dots, \beta_m$, $A$ is a diagonal $m \times m$ matrix and $C$ is an $m \times m$ matrix with
\begin{equation*}
    A_{jj} = \prod_{j' \neq j} (\alpha_j - \alpha_{j'})^{-1}, \quad C_{ij} = \prod_{j' \neq j} (\beta_i - \alpha_{j'})
\end{equation*}
for $i, j \in \{1, \dots, m\}$.
\end{lemma}

\begin{proof}
It is well-known that the entries in the inverse of a Vandermonde matrix can be written as the coefficients of the Lagrange interpolation polynomials $\lambda_j(x)$. In particular, $(W^{-1})_{kj} = \lambda_j^{(k)}$, where
\begin{equation*}
    \lambda_j(x) = \prod_{j' \neq j} \frac{x - \alpha_{j'}}{\alpha_j - \alpha_{j'}} = \sum_{k=1}^m \lambda_j^{(k)} x^{k-1}.
\end{equation*}
Therefore,
\begin{align*}
    (BW^{-1})_{ij} &= \sum_{k=1}^m \lambda_j^{(k)} \beta_i^{k-1} \\
    &= \prod_{j' \neq j} (\beta_i - \alpha_{j'}) \cdot \prod_{j' \neq j} (\alpha_j - \alpha_{j'})^{-1} = (CA)_{ij}.
\end{align*}
Hence, $BW^{-1} = CA$.
\end{proof}

\begin{lemma}\label{lem:large_A_matrix_2_norm_upper_bound}
Let $A$ be as in \cref{lem:vandermonde_inverse_decomposition}. Then, $\lVert A \rVert_2 \leq \frac{2^d}{n}$.
\end{lemma}

\begin{proof}
From the definition of the matrix $A$ and by writing $\alpha_j = \omega_n^{k_j}$ we get that
\begin{align*}
    \lvert A_{jj} \rvert^{-1} &= \big\lvert \prod_{j' \neq j} (\alpha_j - \alpha_{j'}) \big\rvert \\
    &= \big\lvert \prod_{\substack{k \in \mathcal{K} \\ k \neq k_j}} (\omega_n^{k_j} - \omega_n^k) \big\rvert \\
    &= \frac{\big\lvert \prod_{k \neq k_j} (\omega_n^{k_j} - \omega_n^k) \big\rvert}{\big\lvert \prod_{\ell \in \mathcal{L}} (\omega_n^{k_j} - \omega_n^\ell) \big\rvert}.
\end{align*}
Clearly, $\lvert \omega_n^{k_j} - \omega_n^\ell \rvert \leq 2$, so $\big\lvert \prod_{\ell \in \mathcal{L}} (\omega_n^{k_j} - \omega_n^\ell) \big\rvert \leq 2^d$. Furthermore,
\begin{equation*}
    \big\lvert \prod_{k \neq k_j} (\omega_n^{k_j} - \omega_n^k) \big\rvert = \big\lvert \prod_{k \neq k_j} (1 - \omega_n^{k - k_j}) \big\rvert = n
\end{equation*}
as $\prod_{k=1}^{n - 1} (x - \omega_n^k) = 1 + x + \cdots + x^{n - 1}$. Therefore, $\lVert A \rVert_2 = \max_j \lvert A_{jj} \rvert \leq \frac{2^d}{n}$.
\end{proof}

\begin{lemma}\label{lem:large_C_matrix_2_norm_upper_bound}
Let $C$ be as in \cref{lem:vandermonde_inverse_decomposition}. If $\beta_i = \omega_{2mn} \cdot \omega_m^{i-1}$, then $\lVert C \rVert_\mathrm{max} \leq \frac{mn^{d+1}}{2^d d!}$.
\end{lemma}

\begin{proof}
By construction, the element $\beta_i$ is not an $n$th root of unity as
\begin{equation*}
    \beta_i^n = \omega_{2mn}^n \omega_m^{(i-1)n} = \omega_{2m}^1 \omega_{2m}^{2(i-1)n} = \omega_{2m}^{2(i-1)n + 1} \neq 1,
\end{equation*}
since the exponent is not even. By utilizing the factorization $x^n - 1 = \prod_{k=0}^{n-1} (x - \omega_n^k)$, we write
\begin{align*}
    C_{ij} &= \prod_{j' \neq j} (\beta_i - \alpha_{j'}) \\
    &= \prod_{\substack{k \in \mathcal{K} \\ k \neq k_j}} (\beta_i - \omega_n^k) \\
    &= \frac{\beta_i^n - 1}{(\beta_i - \omega_n^{k_j}) \prod_{\ell \in \mathcal{L}} (\beta_i - \omega_n^\ell)}.
\end{align*}
By using the triangle inequality, we find that $|\beta_i^n - 1| \leq |\beta_i^n| + 1 = 2$. Therefore,
\begin{equation*}
    \lvert C_{ij} \rvert \leq \frac{2}{\lvert (\beta_i - \omega_n^{k_j}) \prod_{\ell \in \mathcal{L}} (\beta_i - \omega_n^\ell) \rvert}.
\end{equation*}
Let us focus on he denominator of the above expression. In particular, we need to lower bound the product of distances between $\beta_i$ and any $d + 1$ $n$th roots of unity.

Given two complex numbers $z, z'$ on the unit circle with angle $0 \leq \theta \leq \pi$ between them, the distance between them is $\lvert z - z' \rvert = 2\sin(\theta / 2) \geq 2\theta / \pi$. We may write $\beta_i = \omega_{2mn}^{2(i-1)n + 1}$ and $\omega_n^k = \omega_{2mn}^{2mk}$. Therefore, the angle between $\beta_i$ and the closest $\omega_n^k$ is $\theta \geq \frac{2\pi}{2mn} = \frac{\pi}{mn}$. The smallest distance is then $\geq \frac{2}{mn}$. The second closest $\omega_n^k$ is at an angle of $\theta \geq \frac{\pi}{n}$, since otherwise it would be the closest. Continuing in this fashion, we get that the product is
\begin{equation*}
    \lvert (\beta_i - \omega_n^{k_j}) \prod_{\ell \in \mathcal{L}} (\beta_i - \omega_n^\ell) \rvert \geq \frac{2}{mn} \cdot \frac{2}{n} \cdot \frac{4}{n} \cdots \frac{2d}{n} = \frac{2^{d+1} d!}{mn^{d+1}}. \qedhere
\end{equation*}
\end{proof}

The following theorem states that for a constant $d = n - m$, the condition number of an $m \times m$ Vandermonde matrix whose evaluation points are distinct $n$th roots of unity is polynomial in $n$.

\begin{theorem}\label{thm:large_vandermonde_inverse_2_norm_upper_bound}
Let $W$ be an $m \times m$ Vandermonde matrix whose evaluation points are distinct $n$th roots of unity for $m \leq n$. Then, $\lVert W^{-1} \rVert_2 \leq \frac{1}{d!} n^{d + \sfrac{3}{2}}$, where $d = n - m$. In particular, $\kappa_2(W) \leq \frac{1}{d!} n^{d + 2}$.
\end{theorem}

\begin{proof}
Write $W^{-1} = B^{-1} C A$ by \cref{lem:vandermonde_inverse_decomposition} for $\beta_i = \omega_{2mn} \cdot \omega_m^{i-1}$. Then,
\begin{align*}
    \lVert W^{-1} \rVert_2 &\leq \lVert B^{-1} \rVert_2 \cdot m \lVert C \rVert_\mathrm{max} \cdot \rVert A \rVert_2 \\
    &\leq \frac{1}{\sqrt{m}} \cdot m \frac{mn^{d+1}}{2^d d!} \cdot \frac{2^d}{n} \\
    &= \frac{1}{d!} m^{\tfrac{3}{2}} n^d \\
    &\leq \frac{1}{d!} n^{d + \tfrac{3}{2}}.
\end{align*}
Here we used the fact that the norm is submultiplicative, $\lVert C \rVert_2 \leq m \lVert C \rVert_\mathrm{max}$, $\lVert B^{-1} \rVert_2 = \tfrac{1}{\sqrt{m}}$, and \cref{lem:large_A_matrix_2_norm_upper_bound,lem:large_C_matrix_2_norm_upper_bound}. From the definition of condition number and \cref{lem:vandermonde_2_norm_upper_bound} we obtain
\begin{equation*}
    \kappa_2(W) = \lVert W \rVert_2 \lVert W^{-1} \rVert_2 \leq \sqrt{n} \cdot \frac{1}{d!} n^{d + \tfrac{3}{2}} = \frac{1}{d!} n^{d + 2}. \qedhere
\end{equation*}
\end{proof}

Instead of considering the parameter $d = n - m$ to be constant and looking at asymptotics as $n$ grows, we let $m$ be constant. Consider the sequence defined by the product of the first $m$ entries of the sequence $1, 1, 2, 2, 3, 3, 4, 4, \dots$. It is simple to verify that the $m$th entry is given by $\Pi(m)$, where $\Pi$ is defined by $\Pi(2k) = k!^2$ and $\Pi(2k + 1) = k!^2 \cdot (k + 1)$. This sequence is A010551 on the OEIS \cite{oeis}.

\begin{lemma}\label{lem:small_A_matrix_2_norm_upper_bound}
Let $A$ be as in \cref{lem:vandermonde_inverse_decomposition}. Then, $\lVert A \rVert_2 \leq \frac{n^{m-1}}{4^{m-1} \Pi(m - 1)}$.
\end{lemma}

\begin{proof}
We need to lower bound the product
\begin{equation*}
    \lvert A_{jj} \rvert^{-1} = \prod_{j' \neq j} \lvert \alpha_j - \alpha_{j'} \rvert = \prod_{\substack{k \in \mathcal{K} \\ k \neq k_j}} \lvert \omega_n^{k_j} - \omega_n^k \rvert.
\end{equation*}
The distance between $\omega_n^{k_j}$ and the closest $n$th roots of unity is $\tfrac{4}{n}$. There are two such closest $n$th roots of unity on either side of $\omega_n^{k_j}$, namely $\omega_n^{k_j + 1}$ and $\omega_n^{k_j - 1}$. The next closest $n$th roots of unity are at a distance of $\tfrac{8}{n}$ and there are again two of these. Continuing in this fashion, we obtain
\begin{align*}
    \prod_{\substack{k \in \mathcal{K} \\ k \neq k_j}} \lvert \omega_n^{k_j} - \omega_n^k \rvert &\geq \frac{4}{n} \cdot \frac{4}{n} \cdot \frac{8}{n} \cdot \frac{8}{n} \cdots \\
    &= \frac{4^{m - 1}}{n^{m - 1}} \big(\underbrace{1 \cdot 1 \cdot 2 \cdot 2 \cdots}_{\text{$m - 1$ terms}} \big) \\
    &= \frac{4^{m - 1}}{n^{m - 1}} \Pi(m - 1).
\end{align*}
Therefore, $\lVert A \rVert_2 = \max_j \lvert A_{jj} \rvert \leq \frac{n^{m - 1}}{4^{m - 1} \Pi(m - 1)}$.
\end{proof}

\begin{lemma}\label{lem:small_C_matrix_2_norm_upper_bound}
Let $C$ be as in \cref{lem:vandermonde_inverse_decomposition}. If $|\beta_i| = 1$ for all $i$, then $\lVert C \rVert_\mathrm{max} \leq 2^{m-1}$.
\end{lemma}

\begin{proof}
Clearly, $\lvert \beta_i - \alpha_{j'} \rvert \leq 2$ for all $i, j'$. Thus,
\begin{equation*}
    \lvert C_{ij} \rvert = \prod_{j' \neq j} \lvert \beta_i - \alpha_{j'} \rvert \leq 2^{m - 1}. \qedhere
\end{equation*}
\end{proof}

\begin{theorem}\label{thm:small_vandermonde_inverse_2_norm_upper_bound}
Let $W$ be an $m \times m$ Vandermonde matrix whose evaluation points are distinct $n$th roots of unity for $m \leq n$. Then, $\lVert W^{-1} \rVert_2 \leq \frac{\sqrt{m}}{2^{m - 1} \Pi(m - 1)} n^{m - 1}$.
\end{theorem}

\begin{proof}
Write $W^{-1} = B^{-1} C A$ by \cref{lem:vandermonde_inverse_decomposition} for $\beta_i = \omega_m^{i - 1}$. We can use $\lVert B^{-1} \rVert_2 = \frac{1}{\sqrt{m}}$, and \cref{lem:small_A_matrix_2_norm_upper_bound,lem:large_C_matrix_2_norm_upper_bound} to obtain
\begin{align*}
    \lVert W^{-1} \rVert_2 &\leq \lVert B^{-1} \rVert_2 \cdot \lVert C \rVert_2 \cdot \lVert A \rVert_2 \\
    &\leq \frac{1}{\sqrt{m}} \cdot m \lVert C \rVert_\mathrm{max} \cdot \lVert A \rVert_2 \\
    &\leq \sqrt{m} \cdot 2^{m - 1} \cdot \frac{n^{m - 1}}{4^{m - 1} \Pi(m - 1)} \\
    &= \frac{\sqrt{m}}{2^{m - 1} \Pi(m - 1)} n^{m - 1}. \qedhere
\end{align*}
\end{proof}

\begin{corollary}\label{cor:small_vandermonde_inverse_frobenius_norm_upper_bound}
Let $W$ be an $m \times m$ Vandermonde matrix whose evaluation points are distinct $n$th roots of unity for $m \leq n$. Then, $\lVert W^{-1} \rVert_F \leq \frac{m}{2^{m - 1} \Pi(m - 1)} n^{m - 1}$.
\end{corollary}

For $m = 1$ the above corollary gives $\lVert W^{-1} \rVert_F \leq 1$, which matches the exact value. For $m = 2$ the upper bound is $n$, but it is easy to verify that $\frac{n}{2}$ would be a better bound.

The results in this appendix show that $m \times m$ Vandermonde matrices with $n$th roots of unity as evaluation points have condition numbers that grow polynomially in $n$ when either $m$ or $n - m$ is kept constant. On the other hand, the result is not true if $m = n / 2$ as shown in \cite[Theorem~6.2]{pan2016how}, since in this case the condition numbers grow exponentially in $n$.

\bibliographystyle{ieeetr}
\bibliography{bib}

\end{document}